\documentclass[12pt, draftcls, onecolumn]{IEEEtran}
\usepackage[T1]{fontenc}
\usepackage[utf8]{inputenc}
\usepackage{amssymb}
\usepackage[pdftex]{graphicx}
\usepackage{subcaption} 
\usepackage{amsmath}
\usepackage{amsthm}
\usepackage{mathtools}
\usepackage{array}
\usepackage[noadjust]{cite}
\usepackage{algorithm}
\usepackage{algpseudocode}
\usepackage{xcolor}
\usepackage{tikz}
\usepackage{pgfplots}
\usepackage[hidelinks]{hyperref}
\usepackage[acronym]{glossaries}
\usepackage{dutchcal}
\pgfdeclarelayer{background}
\pgfdeclarelayer{foreground}
\pgfsetlayers{background,main,foreground}
\usepgfplotslibrary{fillbetween}
\usepgfplotslibrary{external}
\usetikzlibrary{spy,backgrounds}
\usepackage[binary-units=true]{siunitx}
\newtheorem{proposition}{Proposition}

\makeatletter
\let\MYcaption\@makecaption
\makeatother

\usepackage{subcaption}

\makeatletter
\let\@makecaption\MYcaption
\makeatother
\algnewcommand{\LineComment}[1]{\(\triangleright\) #1}

\DeclareSIUnit\frame{frame}
\DeclareSIUnit\PRB{PRB}

\newacronym{SNR}{SNR}{signal-to-noise ratio}
\newacronym{LDPC}{LDPC}{low-density parity-check}
\newacronym{QAM}{QAM}{quadrature amplitude modulation}
\newacronym{QPSK}{QPSK}{quadrature phase-shift keying}
\newacronym{BICM}{BICM}{bit-interleaved coded modulation}
\newacronym{BCE}{BCE}{binary cross-entropy}
\newacronym{LLR}{LLR}{log-likelihood ratio}
\newacronym{BMI}{BMI}{bit-wise mutual information}
\newacronym{KL}{KL}{Kullback–Leibler}
\newacronym{BMD}{BMD}{bit-metric decoding}
\newacronym{BER}{BER}{bit error rate}
\newacronym{NN}{NN}{neural network}
\newacronym{GS}{GS}{geometric shaping}
\newacronym{SIP}{SIP}{superimposed pilot}
\newacronym{iid}{i.i.d.\@}{independent and identically distributed}
\newacronym{SGD}{SGD}{stochastic gradient descent}
\newacronym{wrt}{w.r.t.\@}{with respect to}
\newacronym{MAP}{MAP}{maximum a posteriori}
\newacronym{LMMSE}{LMMSE}{linear minimum mean square error}
\newacronym{AWGN}{AWGN}{additive white Gaussian noise}
\newacronym{RBF}{RBF}{Rayleigh block fading}
\newacronym{OFDM}{OFDM}{orthogonal frequency division multiplexing}
\newacronym{3GPP}{3GPP}{3rd Generation Partnership Project}
\newacronym{5GNR}{5G NR}{5G New Radio}
\newacronym{PRB}{PRB}{physical resource block}
\newacronym{IEDD}{IEDD}{iterative estimation, demapping, and decoding}
\newacronym{CSI}{CSI}{channel state information}
\newacronym{MSE}{MSE}{mean squared error}
\newacronym{BPSK}{BPSK}{binary phase-shift keying}
\newacronym{DC}{DC}{direct current}
\newacronym{PAPR}{PAPR}{peak-to-average power ratio}
\newacronym{DMRS}{DMRS}{demodulation reference signal}
\newacronym{TTI}{TTI}{transmission time interval}
\newacronym{CDF}{CDF}{cumulative distribution function}
\newacronym{RE}{RE}{resource element}
\newacronym{MIMO}{MIMO}{multiple-input multiple-output}
\renewcommand{\vec}[1]{\mathbf{#1}}
\newcommand{\vecs}[1]{\boldsymbol{#1}}

\newcommand{\bv}{\vec{b}}

\newcommand{\hv}{\vec{h}}

\newcommand{\mv}{\vec{m}}
\newcommand{\nv}{\vec{n}}

\newcommand{\pv}{\vec{p}}

\newcommand{\wv}{\vec{w}}
\newcommand{\xv}{\vec{x}}
\newcommand{\yv}{\vec{y}}

\newcommand{\zerov}{\vec{0}}

\newcommand{\Am}{\vec{A}}

\newcommand{\Hm}{\vec{H}}
\newcommand{\Id}{\vec{I}}

\newcommand{\Pm}{\vec{P}}

\newcommand{\Rm}{\vec{R}}
\newcommand{\Sm}{\vec{S}}

\newcommand{\Um}{\vec{U}}

\newcommand{\Wm}{\vec{W}}
\newcommand{\Xm}{\vec{X}}
\newcommand{\Ym}{\vec{Y}}

\newcommand{\Lambdam}{\vecs{\Lambda}}

\newcommand{\Bc}{{\cal B}}
\newcommand{\Cc}{{\cal C}}

\newcommand{\Lc}{{\cal L}}

\newcommand{\Nc}{{\cal N}}

\newcommand{\CC}{\mathbb{C}}

\newcommand{\RR}{\mathbb{R}}

\newcommand{\htp}{^{\mathsf{H}}}

\newcommand{\LB}{\left(}
\newcommand{\RB}{\right)}
\newcommand{\LP}{\left\{}
\newcommand{\RP}{\right\}}
\newcommand{\LSB}{\left[}
\newcommand{\RSB}{\right]}

\renewcommand{\ln}[1]{\mathop{\mathrm{ln}}\LB #1\RB}
\renewcommand{\log}[1]{\mathop{\mathrm{log}_2}\LB #1\RB}
\renewcommand{\exp}[1]{\mathop{\mathrm{exp}}\LB #1\RB}

\newcommand{\EE}{{\mathbb{E}}}

\newcommand{\softmax}{\mathop{\mathrm{softmax}}}

 \newcommand{\argmax}[1]{\underset{#1}{\operatorname{arg}\,\operatorname{max}}\;}

\newcommand\abs[1]{\left| #1 \right|}

\begin{document}

\title{End-to-end Learning for OFDM: From Neural Receivers to Pilotless Communication}

\author{Fay\c{c}al Ait Aoudia, \textit{Member, IEEE,} and Jakob Hoydis, \textit{Senior Member, IEEE}%

\thanks{F. Ait Aoudia is with Nokia Bell Labs, 91620 Nozay, France (faycal.ait\_aoudia@nokia-bell-labs.com).
J. Hoydis is with NVIDIA, 06906 Sophia Antipolis, France (jhoydis@nvidia.com). Work carried out while J. Hoydis was with Nokia Bell Labs.}
}

\maketitle

\begin{abstract}
The benefits of end-to-end learning has been demonstrated over AWGN channels but has not yet been quantified over realistic wireless channel models. This work aims to fill this gap by exploring the gains of end-to-end learning over a frequency- and time-selective fading channel using OFDM. With imperfect channel knowledge at the receiver, the shaping gains observed on AWGN channels vanish. Nonetheless, we identify two other sources of performance improvements. The first comes from a neural network-based receiver operating over a large number of subcarriers and OFDM symbols which allows to reduce the number of orthogonal pilots without loss of BER. The second comes from entirely eliminating orthogonal pilots by jointly learning a neural receiver together with either superimposed pilots (SIPs), combined with conventional QAM, or an optimized constellation. The learned constellation works for a wide range of signal-to-noise ratios, Doppler and delay spreads, has zero mean and does hence not contain any form of SIP. Both schemes achieve the same BER as the pilot-based baseline with \SI{7}{\percent} higher throughput. Thus, we believe that a jointly learned transmitter and receiver are a very interesting component for beyond-5G communication systems which could remove the need and associated overhead for demodulation reference signals.

\begin{IEEEkeywords}
Autoencoder, end-to-end learning, geometric shaping, superimposed pilots, orthogonal frequency division multiplexing, frequency-selective fading, channel estimation
\end{IEEEkeywords}

\end{abstract}

\glsresetall
\section{Introduction}
\label{sec:intro}

End-to-end learning has attracted a lot of attention in recent years and it is considered to be a promising technology for future wireless communication systems~\cite{9078454,8839651}. Its key idea is to implement the transmitter, channel, and receiver as a single \gls{NN}, referred to as an autoencoder, that is trained to achieve the highest possible information rate~\cite{9024567, 9118963}.
Since its first application to wireless communications~\cite{8054694}, end-to-end learning has been extended to other fields including optical wireless~\cite{8819929} and optical fiber~\cite{Karanov:18}. However, most of the literature is either simulation-based on simple channel models, such as \gls{AWGN} or \gls{RBF}, or experimental, but performed in static environments~\cite{9024567,8214233}.
Such setups do not account for the Doppler and delay spread encountered in practical wireless systems that lead to variations of the channel response in both time and frequency. The evaluation of end-to-end learning on more realistic  channel models is  overlooked in the existing literature but critical to bring the technology from theory to practice. 

For the reasons mentioned above, we evaluate in this work the benefits of end-to-end learning for \gls{OFDM}-based communication over a time- and frequency-selective fading channel model. We consider a Kronecker structure for the tempo-spectral correlation, using Jakes' Doppler power spectrum and \gls{3GPP} power delay profiles. This model allows us to accurately and conveniently control the Doppler and delay spreads to evaluate the communication performance in different scenarios.
It was already shown in previous work that end-to-end learning enables significant shaping gains on \gls{AWGN} and static channels~\cite{9118963}.
However, we have observed that these gains vanish on realistic channels with imperfect channel knowledge at the receiver.
Therefore, we need to find other ways to increase performance for such channels.
We focus especially on the potential of end-to-end learning to eliminate the need for orthogonal pilot signals, i.e., pilot signals transmitted on dedicated \glspl{RE}, such as \gls{DMRS} in \gls{5GNR}, by jointly optimizing parts of the transmitter and receiver.
Two approaches are presented to achieve this goal.
The first consists in learning \glspl{SIP} that are linearly combined with the \gls{QAM} modulated baseband symbols on the \gls{OFDM} grid. The second approach is more radical and consists in learning a constellation and associated bit labeling which are used to modulate coded bits on all \glspl{RE}. The learned constellation is forced to have zero mean to avoid an unwanted \gls{DC} offset. As a side-effect, it can also not be interpreted as a constellation with superimposed pilots.
In both approaches, the \gls{NN} used at the receiver has a fully convolutional structure and operates jointly on multiple subcarriers and \gls{OFDM} symbols, as in~\cite{8052521,deeprx}.

For benchmarking, we have implemented two strong baseline receivers which rely on \gls{LMMSE} channel estimation with perfect tempo-spectral covariance matrix knowledge and \gls{IEDD}, respectively.
We have also evaluated the performance gains that can be obtained by leveraging only the neural receiver together with standard \gls{QAM} and pilot patterns from \gls{5GNR}.
Our results show that such a receiver enables significant \gls{BER} improvements in scenarios with high mobility and/or when sparse pilot patterns are used, although the channel code is not at all leveraged as for \gls{IEDD}.
These results concur with the observations made in~\cite{8052521,deeprx} regarding the robustness of a neural receiver with respect to a decreased number of pilots.
As a consequence, an \gls{NN}-based receiver with \gls{QAM} and sparse orthogonal pilot patterns from \gls{5GNR} enables wortwhile throughput gains.
Regarding joint optimization of the transmitter and receiver, our results demonstrate that the two approaches we introduce achieve \glspl{BER} similar to the ones achieved by the \gls{NN}-based receiver with \gls{QAM} and orthogonal pilots for various speeds and delay spreads. As a consequence, end-to-end learning enables additional throughput gains in the range from \SI{4}{\percent} to \SI{8}{\percent}, depending on the scenario, by removing all orthogonal pilots.

\subsection*{Related literature}
End-to-end learning over \gls{OFDM} channels was first considered in~\cite{8445920}, in which the autoencoder is assumed to operate over individual \glspl{RE}.
This limitation prevents the receiver from taking advantage of the tempo-spectral correlation of the \gls{OFDM} channel.
This approach is extended towards the learning of a \gls{PAPR} reduction scheme in~\cite{8240644}.
In~\cite{8663458}, one-bit quantization of the received \gls{OFDM} signal is considered, and joint optimization of a precoder and decoder is performed to enable reconstruction of the transmitted signal.

A handful of studies have focused on neural receivers for \gls{OFDM} systems in the recent years.
In~\cite{8052521}, a neural receiver made of dense layers is optimized to jointly process several \gls{OFDM} symbols, with the first symbol carrying pilots.
It is observed that the neural receiver is more robust than conventional receivers when fewer pilots or no cyclic prefix is transmitted.
A large residual convolutional \gls{NN} is considered in~\cite{deeprx} which operates over a large number of subcarriers and \gls{OFDM} symbols.
\glspl{BER} close to what can be achieved with perfect channel knowledge are observed over realistic \gls{3GPP} channel models.
It is also shown that such a neural receiver is more robust against interference compared to traditional approaches.

The goal of our work is to investigate what additional benefits can be achieved from end-to-end learning when a neural receiver is used.
More specifically, the reduction or complete suppression of pilots is explored in this work.
In~\cite{9360873}, the authors develop an end-to-end learning-based joint source and channel coding scheme that does not require pilots.
However, our approach is based on \gls{GS} and \gls{SIP} to achieve pilotless communication for ubiquitous \gls{BICM}.
There also exists a rich literature on this topic, involving schemes that typically require detection algorithms of prohibitive complexity (see, e.g.,~\cite{8012416} and references therein).

The rest of this manuscript is structured as follows: Section~\ref{sec:chmod_bsl} presents the channel model and the receiver baselines against which the \gls{NN}-based approaches are compared.
Section~\ref{sec:e2e} introduces the neural receiver and end-to-end schemes that enable communication without orthogonal pilots.
Section~\ref{sec:results} presents the simulations results while Section~\ref{sec:conclu} concludes the paper.

\paragraph*{Notations}
Boldface upper-case (lower-case) letters denote matrices (column vectors).
$\RR$ ($\CC$) is the set of real (complex) numbers; 
$j$ is the imaginary unit, $()^*$ the complex conjugate operator.
$\ln{\cdot}$ denotes the natural logarithm and $\log{\cdot}$ the binary logarithm.
$\Cc\Nc(\mv,\Sm)$ is the complex multivariate Gaussian distribution with mean $\mv$ and covariance $\Sm$.
The $(i,k)$ element of a matrix $\Xm$ is denoted by $X_{i,k}$. The $k^{th}$ element of a vector $\xv$ is $x_k$  and $\text{diag}(\xv)$ is a diagonal matrix with $\xv$ as diagonal. The operators $()\htp$ and $\text{vec}()$ denote the Hermitian transpose and vectorization, respectively. 
For two matrices $\Xm$ and $\Ym$, their Hadamard and Kronecker products are denoted by $\Xm \circ \Ym$ and $\Xm \otimes \Ym$, respectively.
Random variables are denoted by capital italic font, e.g., $X$, with realizations $x$.
Random vectors are denoted by capital bold calligraphic font, e.g., $\mathbf{\mathcal{X}}$, with realizations $\xv$.
$I(X;Y)$, $p(y|x)$ and $p(x,y)$ represent respectively the mutual information, conditional probability, and joint probability distribution of $X$ and $Y$.

\section{Channel model and receiver baselines}
\label{sec:chmod_bsl}

We start by introducing the channel model which is used throughout this work. We will then detail two baseline receiver algorithms for performance benchmarking with the machine learning-based approaches which are presented in Section~\ref{sec:e2e}.

\subsection{Channel model}
\label{sec:ch_mod}

We consider an \gls{OFDM} system with $n_S$ subcarriers that operates on \glspl{TTI} (or blocks) of $n_T$ consecutive \gls{OFDM} symbols, forming a frame. After cyclic-prefix removal and discrete Fourier transform at the receiver, the corresponding complex-baseband channel model is
\begin{equation}
\label{eq:ch_tr}
\Ym = \Hm \circ \Xm + \Wm
\end{equation}
where $\Ym\in\CC^{n_S \times n_T}$ is the received signal, $\Hm\in\CC^{n_S \times n_T}$ the channel matrix, $\Xm\in\CC^{n_S \times n_T}$ the matrix of transmitted symbols, and $\Wm\in\CC^{n_S \times n_T}$ additive white complex Gaussian noise with variance $\sigma^2$ per element. Depending on the transmission scheme, some elements of $\Xm$ can be used for pilot symbols while the others carry modulated symbols according to a constellation, e.g., \gls{QAM} or a learned constellation. It is also possible that the elements of $\Xm$ carry modulated symbols together with superimposed pilots, as explained later.
For all schemes, the transmitted symbols are assumed to have an average energy equal to one, i.e., $\EE\LP \abs{X_{i,k}}^2 \RP = 1$.
Let us denote by $\hv=\text{vec}(\Hm)\in\CC^{n}$ the vectorization of $\Hm$, where $n = n_S n_T$.
The vector $\hv$ is sampled from a complex Gaussian distribution with zero mean and covariance matrix $\Rm\in\CC^{n\times n}$.
This assumption is valid if the number of propagation paths is large enough and some conditions on the path gains are satisfied~\cite{7155505}.
The matrix $\Rm$ determines the temporal and spectral correlation of the \gls{OFDM} channel. Both are assumed to be separable, i.e., the matrix $\Rm$ can be written as the Kronecker product of a frequency correlation matrix $\mathbf{R_F}\in\CC^{n_S\times n_S}$ and a time correlation matrix $\mathbf{R_T}\in\CC^{n_T\times n_T}$, according to
\begin{equation}
	\Rm = \mathbf{R_F} \otimes \mathbf{R_T}.
\end{equation}
The received signal power is assumed to be uniformly distributed with respect to the angle of arrival at the receiver.
Such a model corresponds to a Clarke-Jakes power angular spectrum which is a reasonable assumption if the receiver is immersed in the propagation clutter~\cite{heath2018foundations}.
The corresponding time correlation matrix is given by~\cite{doi:10.1002/j.1538-7305.1968.tb00069.x}
\begin{equation}
\label{eq:time_cov}
\LSB\mathbf{R_T}\RSB_{i,k} = J_0 \LB 2\pi\frac{v}{c}f_c\Delta_T(i-k) \RB  ,\quad 1\le i,k\le n_T
\end{equation}
where $J_0(\cdot)$ is the zero-order Bessel function of the first kind, $c$ [\si{\meter\per\second}] the speed of light, $v$ [\si{\meter\per\second}] the receiver speed relative to the transmitter, $f_c$ [\si{\hertz}] the carrier frequency and $\Delta_T$ [\si{\second}] the duration of an \gls{OFDM} symbol (including the cycle prefix).
Regarding the frequency correlation, it can be computed given a delay spread $D_s$ [\si{\second}] and a power delay profile composed of $L$ \emph{normalized} delays $\tau_l$ and corresponding powers $S_l$, for $l=1,\dots,L$, according to \cite{heath2018foundations}
\begin{equation}
	\LSB\mathbf{R_F}\RSB_{i,k} = \sum_{l=1}^L S_l e^{j 2 \pi \tau_l Ds \Delta_F (i-k)} ,\quad 1\le i,k\le n_S
\end{equation}
where $\Delta_F$ [\si{\hertz}] is the subcarrier spacing.

Generating channel realizations $\hv$ can be done by sampling from $\Cc\Nc(\zerov, \Rm)$.
This can be achieved by filtering a non-correlated Gaussian vector $\nv \sim \Cc\Nc(\zerov, \Id_n)$ as
\begin{equation}
\hv = \Um\Lambdam^{\frac12} \nv
\end{equation}
where $\Rm = \Um\Lambdam\Um\htp$ is the eigenvalue decomposition of $\Rm$.

\begin{figure}
 	\centering
 	\begin{subfigure}{\linewidth}
 		\centering
		\includegraphics[scale=0.9]{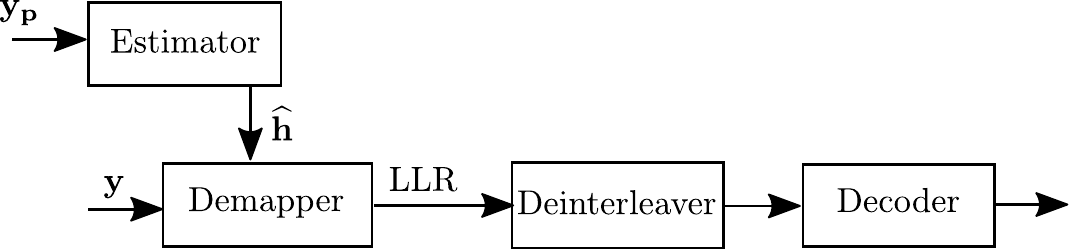}
		\caption{Non-iterative receiver\label{fig:trbsl}}
	\end{subfigure}
	
    \vspace{1.0\baselineskip}
 	\begin{subfigure}{\linewidth}
 		\centering
		\includegraphics[scale=0.9]{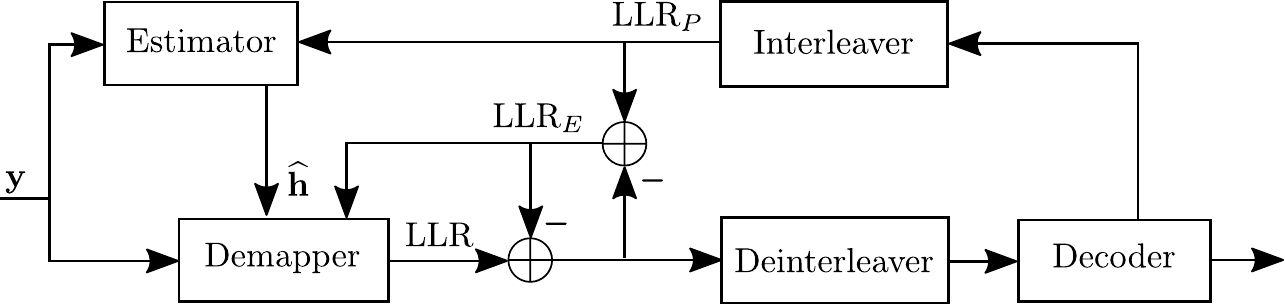}
		\caption{Iterative receiver\label{fig:iedd}}
	\end{subfigure}
	\caption{Considered baselines for the receiver}
\end{figure}

Two receiver baselines are considered in this work.
The first one performs \gls{LMMSE} channel estimation based on transmitted pilots, followed by soft-demapping assuming Gaussian noise, as illustrated in Fig.~\ref{fig:trbsl}.
The second baseline also relies on \gls{LMMSE} channel estimation and Gaussian demapping, but leverages \gls{IEDD}~\cite{1230121}.
These two baselines are detailed in the rest of this section.
They are respectively referred to as the \emph{non-iterative} and \emph{iterative} baselines.

\subsection{LMMSE channel estimation and Gaussian demapping}
\label{sec:non_it}

We denote by $\Pm\in\CC^{n_S \times n_T}$ the pilot matrix, whose entry $P_{i,k}$ is zero if the \gls{RE} on the $i^{th}$ subcarrier and $k^{th}$ time slot is carrying data, or equal to the pilot value otherwise.
Let $n_P$ be the number of pilot-carrying \glspl{RE} which are considered for channel estimation.
For these \glspl{RE}, we can re-write the channel transfer function \eqref{eq:ch_tr} in vectorized form as
\begin{equation}
\mathbf{y_P} = \mathbf{\Pi} \LB \text{diag}(\pv) \hv + \wv\RB
\end{equation}
where $\pv =\text{vec}(\Pm)$, $\wv = \text{vec}(\Wm)$, and $\mathbf{\Pi}$ is the $n_P \times n$ matrix which selects only the elements carrying pilot symbols.
Assuming knowledge of the channel correlation matrix $\Rm$ at the receiver, the \gls{LMMSE} channel estimate is (e.g.,~\cite[Lemma B.17]{massivemimobook}) 
\begin{equation}
\label{eq:lmmse_est}
\widehat{\hv} = \Rm \text{diag}(\pv)\htp \mathbf{\Pi}\htp \LB \mathbf{\Pi} \LB \text{diag}(\pv) \Rm \text{diag}(\pv)\htp + \sigma^2 \Id_n \RB \mathbf{\Pi}\htp \RB^{-1} \mathbf{y_p}.
\end{equation}
Using this result, the received signal in \eqref{eq:ch_tr}, which incorporates both pilots and data, can be re-written in vector form as
\begin{equation}
\yv = \text{diag}(\widehat{\hv})\xv + \underbrace{\text{diag}(\widetilde{\hv})\xv + \wv}_{ \coloneqq \widetilde{\wv}}
\end{equation}
where $\widetilde{\hv} = \hv - \widehat{\hv}$ is the channel estimation error with correlation matrix $\widetilde{\Rm}\in\CC^{n\times n}$, given as
\begin{equation}
\label{eq:est_err}
\widetilde{\Rm} = \EE \LP \widetilde{\hv}\widetilde{\hv}\htp \RP = \Rm - \Rm \text{diag}(\pv)\htp \mathbf{\Pi}\htp \LB \mathbf{\Pi} \LB \text{diag}(\pv) \Rm \text{diag}(\pv)\htp + \sigma^2 \Id_n \RB \mathbf{\Pi}\htp \RB^{-1} \mathbf{\Pi} \text{diag}(\pv)\Rm
\end{equation}
and $\widetilde{\wv}$ is the sum of noise and residual interference due to imperfect channel estimation.

Soft-demapping is performed assuming that $\widetilde{\wv}$ is Gaussian.\footnote{This is typically not true as $\text{diag}(\widetilde{\hv})\xv$ is not Gaussian.}
Let us denote by $m$ the number of bits per channel use, by $\Cc = \{c_1,\dots,c_{2^m}\}$ the constellation, and by $\Cc_{i,0}$ ($\Cc_{i,1}$) the subset of $\Cc$ which contains all constellation points with the $i^{th}$ bit label set to 0 (1).
The \gls{LLR} for the $i^{th}$ bit ($i \in \{1,\dots,m\}$) of the $k^{th}$ \gls{RE} ($k \in \{1,\dots,n\}$) is computed as follows:
\begin{equation}
\text{LLR}(k,i) = \ln{\frac{\sum_{c \in \Cc_{i,1}} \exp{-\frac{1}{\widetilde{\sigma}_k^2} \abs{y_k - \widehat{h}_k c}^2}}{\sum_{c \in \Cc_{i,0}} \exp{-\frac{1}{\widetilde{\sigma}_k^2} \abs{y_k - \widehat{h}_kc}^2}}}
\end{equation}
where $\widetilde{\sigma}_k^2 = \EE \LP \widetilde{w}_k\widetilde{w}_k^* \RP = \widetilde{R}_{k,k} + \sigma^2$  for $k=1,\dots, n$. Whenever $k$ corresponds to the index of a pilot symbol, no \gls{LLR} value is computed. After deinterleaving, the \glspl{LLR} are fed to a channel decoding algorithm (e.g., belief propagation) which computes predictions of the transmitted bits.

\subsection{Iterative channel estimation, demapping and decoding}
\label{sec:iedd}

The key idea of \gls{IEDD} is to leverage the channel code to improve channel estimation and demapping.
Instead of running the estimator, demapper, and decoder once, as illustrated in Fig.~\ref{fig:trbsl}, \gls{IEDD} consists of running the estimator, demapper, and a few iterations of the decoder in an iterative manner, as illustrated in Fig.~\ref{fig:iedd}.
At each iteration, the \glspl{LLR} generated by the decoder are used as prior knowledge on the transmitted bits by the estimator and demapper.
Formally, at each iteration, prior information is assumed to be available to the channel estimator for each bit $i \in \{1,\dots,m\}$ of each \gls{RE} $k \in \{1,\dots,n\}$ in the form of \glspl{LLR} denoted by $\text{LLR}_P(k,i)$.
From these \glspl{LLR}, a prior distribution $P_{X_k}$ is approximated for each transmitted data symbol $x_k\in\Cc$ as follows:
\begin{equation}
\label{eq:iedd_prs}
\LSB P_{X_k}(c_1),\dots,P_{X_k}(c_{2^m}) \RSB \approx \text{softmax}\LB \sum_{i=1}^m c^{(i)}_1 \text{LLR}_P(k,i), \dots, \sum_{i=1}^m c^{(i)}_{2^m} \text{LLR}_P(k,i) \RB
\end{equation}
where $c^{(i)}_u\in\{0,1\}$ refers to the value of the $i^{th}$ bit label of the constellation symbol $c_u$, and the softmax function is defined as $\softmax(l_1,\dots,l_{2^m}) = \LSB \frac{\exp{l_1}}{\sum_{i=1}^{2^m} \exp{l_i}},\dots,\frac{\exp{l_{2^m}}}{\sum_{i=1}^{2^m} \exp{l_i}} \RSB$.
A derivation of all the equations related to the \gls{IEDD} baseline can be found in Appendix\ref{app:iedd_der}.
The prior distribution is over the constellation set $\Cc$ from which $x_k$ takes its value.
At the first iteration, no prior information is available and the prior \glspl{LLR} are therefore set to zero.
Regarding the pilot symbols, the prior distribution always corresponds to the deterministic distribution with probability one for the pilot symbol and zero for any other symbol, i.e., if the symbol $p \in \CC$ is transmitted as a reference signal on the $k^{th}$ \gls{RE}, then $P_{X_k}(X_k = p) = 1$ and $P_{X_k}(X_k \neq p) = 0$.

Given the prior distributions $P_{X_k}$, the \gls{LMMSE} channel estimate is
\begin{equation}
\label{eq:iedd_est}
\widehat{\hv}' = \Rm\text{diag}\LB\bar{\xv}\RB\htp \LB \Rm \circ \EE \LP \xv \xv\htp \RP  + \sigma^2\Id \RB^{-1}\yv
\end{equation}
where $\bar{x}_k \coloneqq \EE_{P_{X_k}}\LP x_k \RP$ and
\begin{equation}
\EE \LP \xv \xv\htp \RP_{i,k} =
\begin{cases}
\EE_{P_{X_k}} \LP |x_k|^2 \RP & \text{if } i = k\\
\EE_{P_{X_i}} \LP x_i \RP \EE_{P_{X_k}} \LP x_k^* \RP & \text{otherwise,}
\end{cases}
\end{equation}
which assumes that the symbols forming $\xv$ are not correlated.\footnote{This might not be true in practice as the channel code introduces redundancy.}
Equation \eqref{eq:iedd_est} is a form of \emph{data-aided channel estimation} as prior information on the data symbols is used in addition to the pilot symbols to perform channel estimation.
The correlation matrix of the estimation error is
\begin{equation}
\label{eq:est_err_iedd}
\widetilde{\Rm}' = \Rm - \Rm\text{diag}\LB\bar{\xv}\RB\htp \LB \Rm \circ \EE \LP \xv \xv\htp \RP  + \sigma^2\Id \RB^{-1} \text{diag}\LB\bar{\xv}\RB\Rm.
\end{equation}
Demapping leverages the \emph{extrinsic information} generated by the decoder, which can be intuitively seen as the additional information generated by the decoder only, and is obtained by subtracting its input from its output~\cite{1230121}, as illustrated in Fig.~\ref{fig:iedd}.
The extrinsic information is available in the form of \glspl{LLR}, which is denoted by $\text{LLR}_E(k,i)$ for the $i^{th}$ bit transmitted in the $k^{th}$ \gls{RE}.
The \gls{LLR} for the $i^{th}$ bit of the $k^{th}$ \gls{RE} computed by the demapper is
\begin{equation}
\label{eq:iedd_llr}
\text{LLR}(k,i) = \ln{\frac{\sum_{c \in \Cc_{i,1}} \exp{-\frac{1}{\widetilde{\sigma_k}'^2} \abs{y_k - \widehat{h}_k c}^2 + \sum_{l=1}^m c^{(l)}\text{LLR}_E(k,l)}}{\sum_{c \in \Cc_{i,0}} \exp{-\frac{1}{\widetilde{\sigma_k}'^2} \abs{y_k - \widehat{h}_kc}^2 + \sum_{l=1}^m c^{(l)}\text{LLR}_E(k,l)}}}
\end{equation}
where $\widetilde{\sigma_k}'^2 = \widetilde{R'}_{k,k} + \sigma^2$, for $k=1,\dots, n$.
As shown in Fig.~\ref{fig:iedd}, only the extrinsic information generated by the demapper is forwarded to the decoder, which is obtained by subtracting from the \glspl{LLR}~(\ref{eq:iedd_llr}) the prior information $\text{LLR}_P$.

\section{End-to-end learning for OFDM}
\label{sec:e2e}

\begin{figure}
 	\centering
	\includegraphics[scale=0.8]{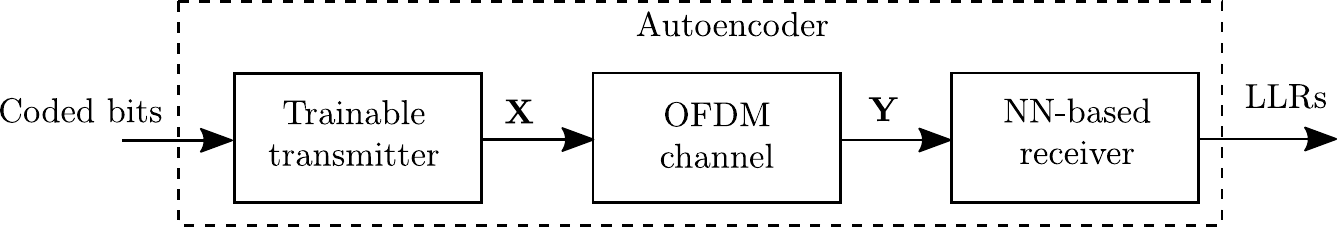}
	\caption{Autoencoder-based communication system\label{fig:e2e_ovw}}
\end{figure}

End-to-end learning of communication systems~\cite{8054694} consists in implementing a transmitter, channel, and receiver as a single \gls{NN} referred to as an autoencoder, as illustrated in Fig.~\ref{fig:e2e_ovw}, and jointly optimizing the trainable parameters of the transmitter and receiver for a specific channel model.
In our setting, training aims at minimizing the total \gls{BCE} [\si{\bit\per\frame}], defined as
\begin{equation}
\label{eq:bce}
\Lc \coloneqq -\sum_{k \in \Nc_D}\sum_{i=1}^m \EE_{b_{k,i}, \yv} \LP \log{Q_{k,i}\LB b_{k,i}|\yv \RB }\RP
\end{equation}
where $\Nc_D$ is the set of size $n_D$ of indexes of \glspl{RE} carrying data symbols, $b_{k,i}$ is the $i^{th}$ bit transmitted in the $k^{th}$ resource element, and $Q_{k,i}\LB \cdot | \yv \RB$ is the receiver estimate of the posterior distribution on the $i^{th}$ bit transmitted in the $k^{th}$ \gls{RE} given the channel output $\yv$.
Since \eqref{eq:bce} is numerically difficult to compute, it is estimated through Monte Carlo sampling as
\begin{equation}
\label{eq:bce_mc}
\Lc \approx -\frac{1}{S} \sum_{l=1}^S \sum_{k \in \Nc_D}\sum_{i=1}^m \LP \log{Q_{k,i}\LB b_{k,i}^{[l]}|\yv^{[l]} \RB }\RP
\end{equation}
where $S$ is the batch size, i.e., the number of samples used to estimate $\Lc$, and the superscript $[l]$ is used to refer to the $l^{th}$ sample forming a batch.
The total \gls{BCE}~(\ref{eq:bce}) can be rewritten as
\begin{equation}
\Lc = n_Dm - R
\end{equation}
with $R$ defined as
\begin{equation}
\label{eq:rate}
R \coloneqq \sum_{k \in \Nc_D}\sum_{i=1}^m I(B_{k,i};\mathbf{\mathcal{Y}}) - \sum_{k \in \Nc_D}\sum_{i=1}^m \EE_{\yv}\LP \text{D}_{\text{KL}}\LB P_{B_{k,i}}\LB \cdot|\yv \RB || Q_{k,i}\LB \cdot | \yv \RB  \RB \RP
\end{equation}
where $B_{k,i}$ is the random variable corresponding to the $i^{th}$ transmitted bit in the $k^{th}$ \gls{RE}, $\mathbf{\mathcal{Y}}$ is the random variable corresponding to the received symbols, $\text{D}_{\text{KL}} \LB \cdot || \cdot \RB$ is the \gls{KL} divergence, and $P_{B_{k,i}}\LB \cdot|\yv \RB$ is the \emph{true} posterior distribution on the transmitted bits given the channel output $\yv$.
The first term in~(\ref{eq:rate}) is the maximum information rate that can be achieved assuming an ideal \gls{BMD} receiver and depends only on the transmitter and channel.
The second term in~(\ref{eq:rate}) is the rate loss caused by an imperfect receiver, which equals zero if the receiver computes the true posterior distribution of the transmitted bits.
Interestingly, in accordance to the following proposition, $R$ is an achievable rate assuming a mismatched \gls{BMD} receiver is used.
Therefore, by minimizing $\Lc$, one actually maximizes the achievable rate $R$.
\begin{proposition} \label{prop:ach}
$R$ is an achievable rate in the sense of the standard definition~\cite[§8.5]{cover1999elements}.
\end{proposition}
\begin{proof}
See Appendix\ref{app:rate_proof}.
\end{proof}

Training of the transmitter typically involves joint optimization of the constellation geometry and bit labeling~\cite{9118963}.
Such an approach is adopted in this work and extended to \gls{OFDM} channels in Section~\ref{sec:e2e_gs}.
Moreover, a centered constellation is learned, assuming that no orthogonal pilots are transmitted.
This avoids an unwanted \gls{DC} offset and removes the throughput loss due to the transmission of reference signals that carry no data.
Motivated by previous observations that training of end-to-end systems over fading channels leads to the learning of constellations that embed \glspl{SIP}~\cite{8054694,8792076}, we also consider the learning of an \gls{SIP} scheme based on conventional \gls{QAM} in Section~\ref{sec:e2e_sip}.
The architecture of the \gls{NN}-based receiver is detailed in Section~\ref{sec:nn_arch}.

\subsection{Learning of geometric shaping and bit-labeling}
\label{sec:e2e_gs}

\begin{figure}
 	\centering
	\includegraphics[scale=0.8]{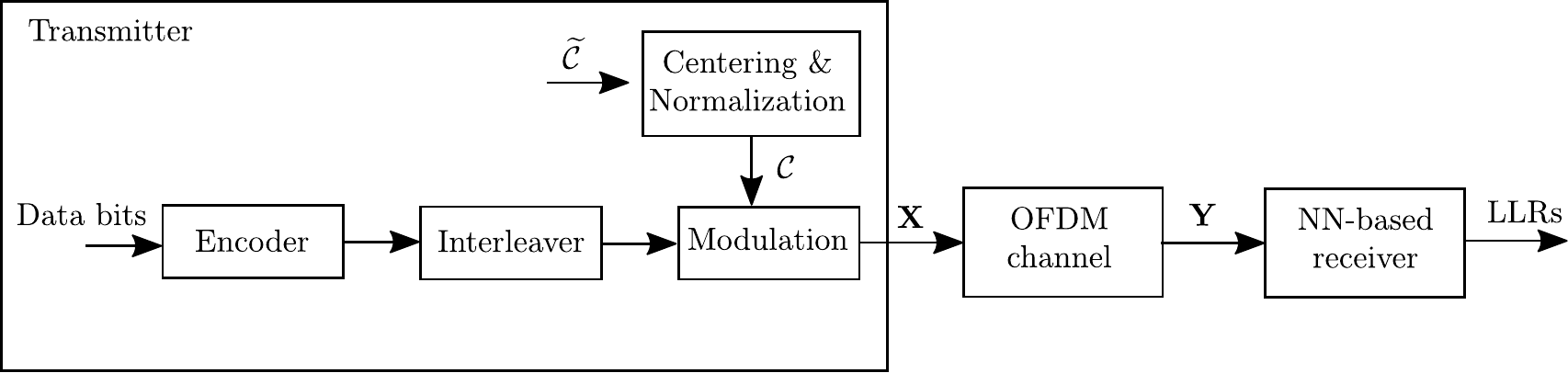}
	\caption{End-to-end system with trainable constellation and receiver\label{fig:gs}}
\end{figure}

The system architecture we have adopted in this work is depicted in Fig.~\ref{fig:gs}.
On the transmitter side, the trainable parameters consist of a set denoted by $\widetilde{\Cc}$ of $2^m$ complex numbers corresponding to the constellation points.
The constellation used for transmitting data is obtained by centering and normalizing $\widetilde{\Cc}$, i.e.,
\begin{equation}
\label{eq:const}
\Cc = \frac{\widetilde{\Cc} - \frac{1}{2^m}\sum_{c \in \widetilde{\Cc}}c}{\sqrt{\frac{1}{2^m}\sum_{c \in \widetilde{\Cc}}\abs{c}^2 - \abs{\frac{1}{2^m}\sum_{c \in \widetilde{\Cc}}c}^2}}.
\end{equation}
Normalization of the constellation ensures it has unit average power, while  centering forces the constellation to have zero mean and therefore avoids an undesired \gls{DC} offset.
Centering the constellation as done in~(\ref{eq:const}) prevents learning of embedded \glspl{SIP}.
As no orthogonal pilots are transmitted, the receiver can only exploit the constellation geometry to reconstruct the transmitted bits.
Compared to previous work such as~\cite{9118963}, a single constellation is learned, which is used for all \glspl{SNR}, Doppler, and delay spreads.
This constellation is learned offline, prior to deployment, and used in place of conventional geometries, such as \gls{QAM}.
Therefore, this approach leads to no additional complexity on the transmitter side over traditional systems.
On the receiver side, an \gls{NN} that operates on multiple subcarriers and \gls{OFDM} symbols is leveraged, whose architecture is detailed in Section~\ref{sec:nn_arch}.
If using an \gls{NN} to perform detection might seem prohibitive because of the apparent incurred complexity, such algorithms highly benefit from hardware acceleration, which might lead to drastic energy and latency reduction~\cite{10.1145/3361682}.
As in~\cite{9118963}, training of the end-to-end systems is done on the total \gls{BCE} estimated by~(\ref{eq:bce_mc}).

\subsection{Learning of superimposed pilots}
\label{sec:e2e_sip}

Typical communication systems rely on pilot patterns which describe how pilots and data symbols are scattered over the \glspl{RE} within a frame.
Such pilots are said to be orthogonal as each \gls{RE} is either carrying a modulated data symbol or a reference signal.
\glspl{SIP} adopt a different approach by splitting the energy available for each \gls{RE} between a reference signal and a data symbol.
Therefore, compared to orthogonal pilots, all \glspl{RE} are used to transmit modulated data symbols.
Formally, the transmitted symbol on the $i^{th}$ subcarrier and on the $k^{th}$ time slot is
\begin{equation}
\label{eq:sip_elem}
X_{i,k} = \underbrace{\sqrt{1-A_{i,k}}\widetilde{X}_{i,k}}_{\text{Data}} + \underbrace{\sqrt{A_{i,k}}P_{i,k}}_{\text{Pilot}}
\end{equation}
where $\widetilde{X}_{i,k}$ is the modulated data symbol that takes values in a constellation set $\Cc$ with zero mean, $A_{i,k} \in [0,1]$ controls the fraction of energy allocated to the reference signal (the rest being allocated to the data symbol), and $P_{i,k}$ is a predefined reference signal.
Equation~(\ref{eq:sip_elem}) can be written in matrix form as
\begin{equation}
\Xm = \sqrt{\mathbf{1}-\Am} \circ \widetilde{\Xm} + \sqrt{\Am}\circ\Pm
\end{equation}
where $\mathbf{1}$ is the matrix full of ones and the square root is taken elementwisely.

We now propose to jointly optimize the pilot allocation matrix $\Am$ and the \gls{NN}-based receiver to maximize the achievable rate~(\ref{eq:rate}).
Our approach differs from previous work on deep learning of pilot patterns~\cite{pp1,pp2,9037126} in several ways.
Firstly, we focus on \glspl{SIP}, whereas prior art focuses on orthogonal pilots.
Secondly, the referenced papers consider learning of a channel estimator to reduce the \gls{MSE} between the true channel and the estimate.
In this work, we train on the achievable rate~(\ref{eq:rate}) as this is the metric of interest when optimizing communication systems from end-to-end.
That way, the system is optimized to learn the right amount of energy which should be allocated to pilots, whereas previous approaches required a pruning step and  regularization as training on the \gls{MSE} leads to most of the energy being allocated to pilots, as the system is not inclined to free resources for data transmission.

\begin{figure}
 	\centering
	\includegraphics[scale=0.8]{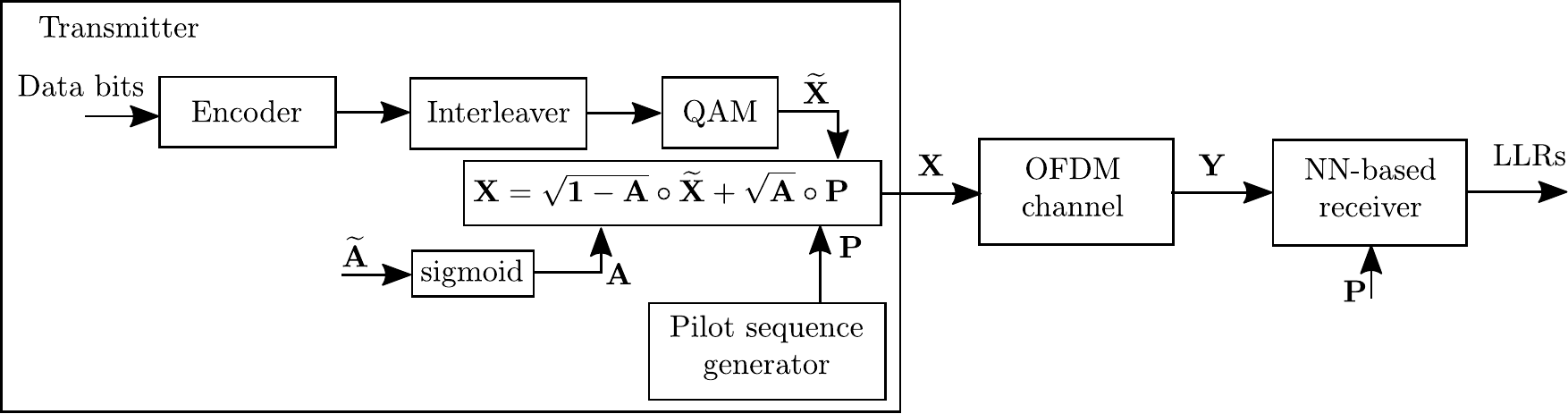}
	\caption{QAM with superimposed pilots\label{fig:qam_sip}}
\end{figure}

The considered system is shown in Fig.~\ref{fig:qam_sip}.
The pilot allocation matrix $\Am$ is obtained by elementwisely taking the sigmoid of an unconstrained matrix $\widetilde{\Am}$, whose coefficients are optimized at training.
To ensure that the transmitted signal has zero average energy, the pilot matrix $\Pm$ is formed by randomly sampling a \gls{BPSK} constellation according to a pseudo-random sequence with zero mean.
As shown in Fig.~\ref{fig:qam_sip}, the so-generated pilot matrix is fed as an additional input to the \gls{NN}-based receiver, whose architecture is detailed thereafter.

\subsection{Receiver architecture}
\label{sec:nn_arch}

\begin{figure}
 	\centering
 	\begin{subfigure}{\linewidth}
 		\centering
		\includegraphics[scale=0.9]{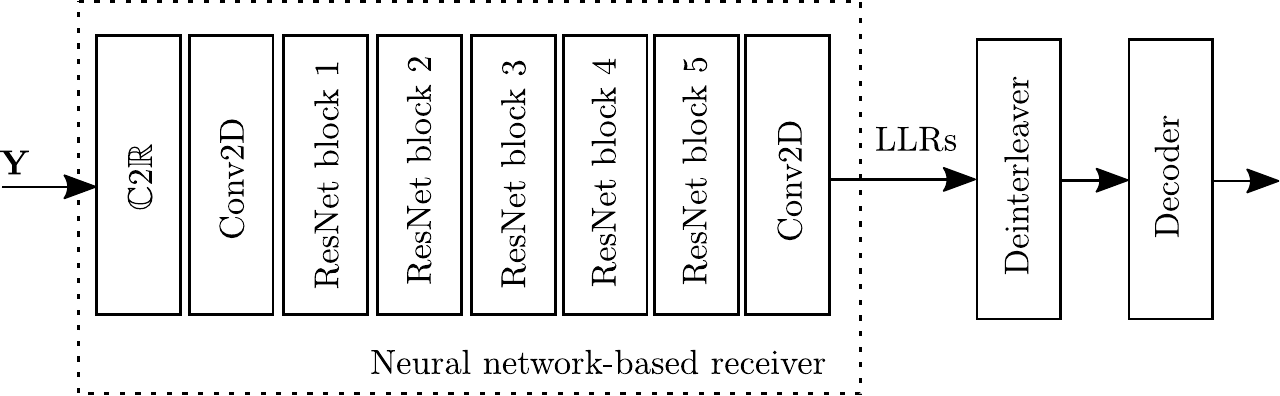}
		\caption{Receiver architecture}
	\end{subfigure}
	
    \vspace{1.0\baselineskip}
 	\begin{subfigure}{\linewidth}
 		\centering
		\includegraphics[scale=0.9]{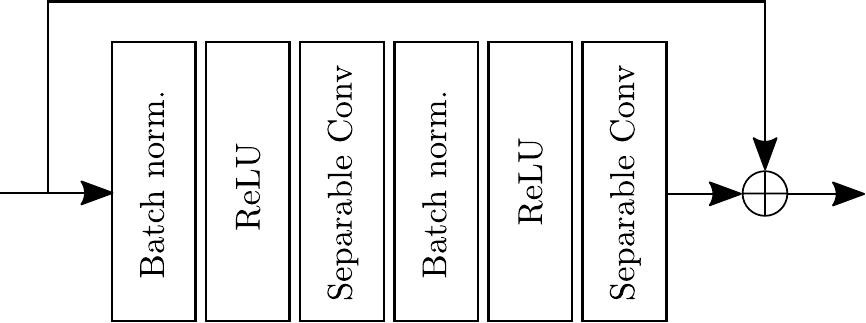}
		\caption{ResNet block\label{fig:resnet}}
	\end{subfigure}
	\caption{Architecture of the \gls{NN}-based receiver\label{fig:nn_rx}}
\end{figure}

The architecture of the \gls{NN} implementing the receiver is shown in Fig.~\ref{fig:nn_rx}.
It is a convolutional residual \gls{NN}~\cite{He_2016_CVPR} that takes as input the received baseband channel samples $\Ym$ of dimension $n_S \times n_T$, and outputs a 3-dimensional tensor of \glspl{LLR} of dimension $n_S \times n_T \times m$ that is fed to the channel decoder after a deinterleaving step.
The \gls{NN} hence substitutes the estimator and demapper shown in Fig.~\ref{fig:trbsl}.
The first layer $\CC2\RR$ converts the complex-valued input tensor of dimension $n_S \times n_T$ into a real-valued 3-dimensional tensor of dimension $n_S \times n_T \times 2$ by stacking the real and imaginary parts into an additional dimension.
No iterative scheme is involved as in the \gls{IEDD} baseline.
Separable convolutional layers are used to reduce the number of weights, without incurring significant loss of performance.
Table~\ref{tab:nn_rx} provides details on the \gls{NN} implementing the receiver.
All convolutional layers use zero-padding to ensure that the dimensions of the output are the same as the ones of the input.
Dilation is leveraged to increase the receptive field of the convolutional layers.
The two separable convolutional layers in a ResNet block (Fig.~\ref{fig:resnet}) share the same dimensions.
A similar but somewhat larger architecture was used in~\cite{deeprx} with great success.

\begin{table}
\begin{center}
  \begin{tabular}{ l | c | c | c }
    \hline
    Layer			& Channels	& Kernel size 	& Dilatation rate	\\ \hline
    Input Conv2D 	& 64		& (3,3)			& (1,1)				\\ \hline
    ResNet block 1	& 128		& (7,7)			& (7,2)				\\ \hline
    ResNet block 2	& 128		& (7,5)			& (7,1)				\\ \hline
    ResNet block 3	& 128		& (5,3)			& (1,2)				\\ \hline
    ResNet block 4	& 128		& (3,3)			& (1,1)				\\ \hline
    ResNet block 5	& 128		& (3,3)			& (1,1)				\\ \hline
    Output Conv2D 	& $m$		& (1,1)			& (1,1)				\\ \hline
  \end{tabular}
\end{center}
\caption{Architecture details of the \gls{NN} implementing the receiver\label{tab:nn_rx}}
\end{table}

\section{Simulation results}
\label{sec:results}

We will now present the results of the simulations we have conducted to evaluate the machine learning-based schemes introduced in the previous section.
We start by explaining the training and evaluation setup.
Then, we benchmark the \gls{NN}-based receiver introduced in Section~\ref{sec:nn_arch} against the non-iterative and iterative baselines presented in Section~\ref{sec:chmod_bsl}, considering conventional \gls{QAM} at the transmitter with pilots patterns from \gls{5GNR} and over different speeds.
The end-to-end learning schemes introduced in Sections~\ref{sec:e2e_gs} and~\ref{sec:e2e_sip} are then compared to conventional approaches leveraging orthogonal pilot patterns.

\subsection{Evaluation setup}
\label{sec:setup}

\begin{figure}
 	\centering
 	\begin{subfigure}{0.45\linewidth}
 		\centering
		\includegraphics[scale=1.0]{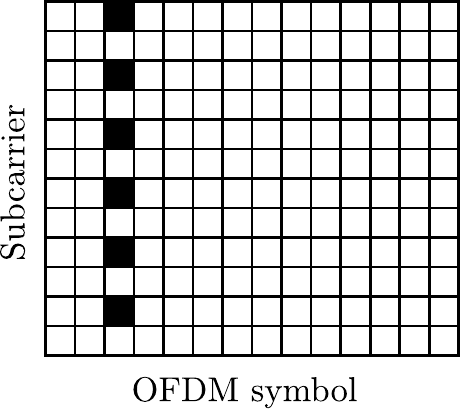}
		\caption{Pilot pattern ``1P''\label{fig:p1}}
	\end{subfigure}\quad\quad
 	\begin{subfigure}{0.45\linewidth}
 		\centering
		\includegraphics[scale=1.0]{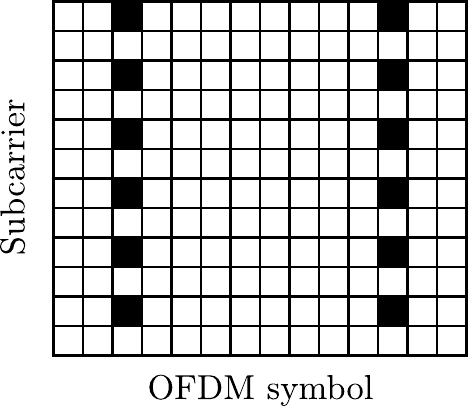}
		\caption{Pilot pattern ``2P''\label{fig:p2}}
	\end{subfigure}
	\caption{Pilot patterns from \gls{5GNR} used in simulations\label{fig:pilot}}
\end{figure}

\begin{table}
\begin{center}
  \begin{tabular}{ l | c | c }
    \hline
    Parameter			& Symbol (if any) & Value			\\ \hline
    Number of OFDM symbols		&	$n_T$ 		& 14 (1 slot)	\\ \hline
    Number of subcarriers 		&	$n_S$		& 72 (6 PRBs)	\\ \hline
    Frequency carrier 			&	$f_c$		& \SI{2.6}{\GHz}	\\ \hline
    Subcarrier spacing			&	$\Delta_f$	& \SI{15}{\kHz}	\\ \hline
    Cycle prefix duration		&	(None)	& \SI{5.2}{\micro\s}	\\ \hline
    Noise variance				&	$\sigma^2$	& \SI{-20}{\deci\bel}		\\ \hline
    Learning rate				&	(None)		& $10^{-3}$		\\ \hline
    Batch size for training		&	$S$			& 100 frames	\\ \hline
    Power delay profiles for training		&	(None)		& TDL-B and TDL-C \\ \hline
    Power delay profiles for evaluation		&	(None)		& TDL-A \\ \hline
    Bit per channel use			&	$m$		&	\SI{6}{\bit}	\\ \hline
    Code length					&	(None)		& \SI{1944}{\bit} \\ \hline
    Code rate					&	$r$			& $\frac{2}{3}$ \\ \hline
    Low speed range				&	(None)		& $0$ to \SI{5.1}{\m\per\s}	\\ \hline
    Medium speed range			&	(None)		& $13.6$ to \SI{18.8}{\m\per\s}	\\ \hline
    High speed range			&	(None)		& $27.4$ to \SI{32.5}{\m\per\s}	\\ \hline
    Speed range used for training	&	(None)		& $0$ to \SI{32.5}{\m\per\s}	\\ \hline
    Delay spread range for evaluation	&	(None)		& $70$ to \SI{140}{\ns}	\\ \hline
    Delay spread range used for training	&	(None)		& $10$ to \SI{1000}{\ns}	\\ \hline
  \end{tabular}
\end{center}
\caption{Parameters used for training and evaluation\label{tab:setup_param}}
\end{table}

The channel model introduced in Section~\ref{sec:ch_mod} was considered to train the end-to-end systems and to evaluate them against multiple baselines.
In addition to the non-iterative and iterative baselines, perfect channel knowledge at the receiver was also considered as it provides an upper-bound on the achievable performance.
Table~\ref{tab:setup_param} shows the parameters used to train and evaluate the machine learning-based approaches.
Three speed ranges were considered, referred to as ``low'', ``medium'', and ``high'' that are within the range from $0$ to \SI{35}{\m\per\s}.
Training of the machine learning-based approaches was done over the entire speed range and over the delay spread range from $10$ to \SI{1000}{\ns}.
The TDL-B and TDL-C \gls{3GPP} power delay profiles were used for training, whereas the TDL-A profile was used for evaluation to ensure that no overfitting to a particular power delay profile has occurred.
A Gray-labeled \gls{QAM} with modulation order 64 ($m = 6$) was used.
For a given \gls{OFDM} frame with realization $\Hm$, the energy per bit to noise power spectral density ratio is defined as
\begin{equation}
\frac{E_b}{\sigma^2} \coloneqq \frac{1}{\rho}\frac{\sum_{i=1}^{n_S} \sum_{k=1}^{n_T} \abs{H_{i,k}}^2}{nmr\sigma^2}
\end{equation}
assuming $\EE\LP \abs{X_{i,k}}^2 \RP = 1$ for all \glspl{RE} $(i,k)$, and where $\rho$ is the ratio of \glspl{RE} carrying data symbols within a frame (the rest of the \glspl{RE} being used for pilots) and $r$ is the code rate.
Two pilot patterns from \gls{5GNR} were considered, which are shown in Fig.~\ref{fig:pilot}.
The first one has pilots on only one \gls{OFDM} symbol (Fig.~\ref{fig:p1}, $\rho = \frac{162}{168}$), whereas the second one has extra pilots on a second \gls{OFDM} symbol (Fig.~\ref{fig:p2}, $\rho = \frac{156}{168}$), which makes it more suitable to high mobility scenarios.
These two pilot patterns are referred to as ``1P'' and ``2P'', respectively.

A standard IEEE~802.11n \gls{LDPC} code of length $1944\:$bit and with rate $\frac{2}{3}$ was used~\cite{5307322}.
Decoding was done with conventional belief-propagation, using 40 iterations for the non-iterative baseline and machine learning-based approaches.
Regarding the iterative baseline, 4 iterations were performed, each involving a channel estimation and soft-demapping step as well as 10 iterations of belief-propagation.
Each transmitted \gls{OFDM} frame contained 3 codewords and was filled up with randomly generated padding bits.
Interleaving was performed within individual frames.

The \gls{NN}-based receiver operates on the entire frame.
For fairness, the channel estimations~(\ref{eq:lmmse_est}) and~(\ref{eq:iedd_est}) for the non-iterative and iterative baselines, respectively, are also performed on the entire frame.
This involves the multiplication and inversion of matrices of dimension $1008\times1008$.
Moreover, it requires knowledge of the channel correlation matrix $\Rm$, which depends on the the Doppler and delay spread.
However, the receiver typically does not have access to this information, and assuming it to be available would lead to an unfair comparison with the \gls{NN}-based receiver which is not fed with the Doppler and delay spread.
Therefore, a unique correlation matrix has been estimated by sampling $10^6$ random channel realizations over the entire speed and delay spread ranges, and using the TDL-B and TDL-C power delay profiles, i.e., $\Rm$ was approximated by $\frac{1}{10^6} \sum_{l=1}^{10^6} \hv_{l} \hv_{l}\htp$, where $\hv_l$ is the $l^{th}$ sampled channel realization. The so-obtained correlation matrix estimate was used to implement the baselines.

\subsection{Evaluation of the NN-based receiver}
\label{sec:rx_eval}

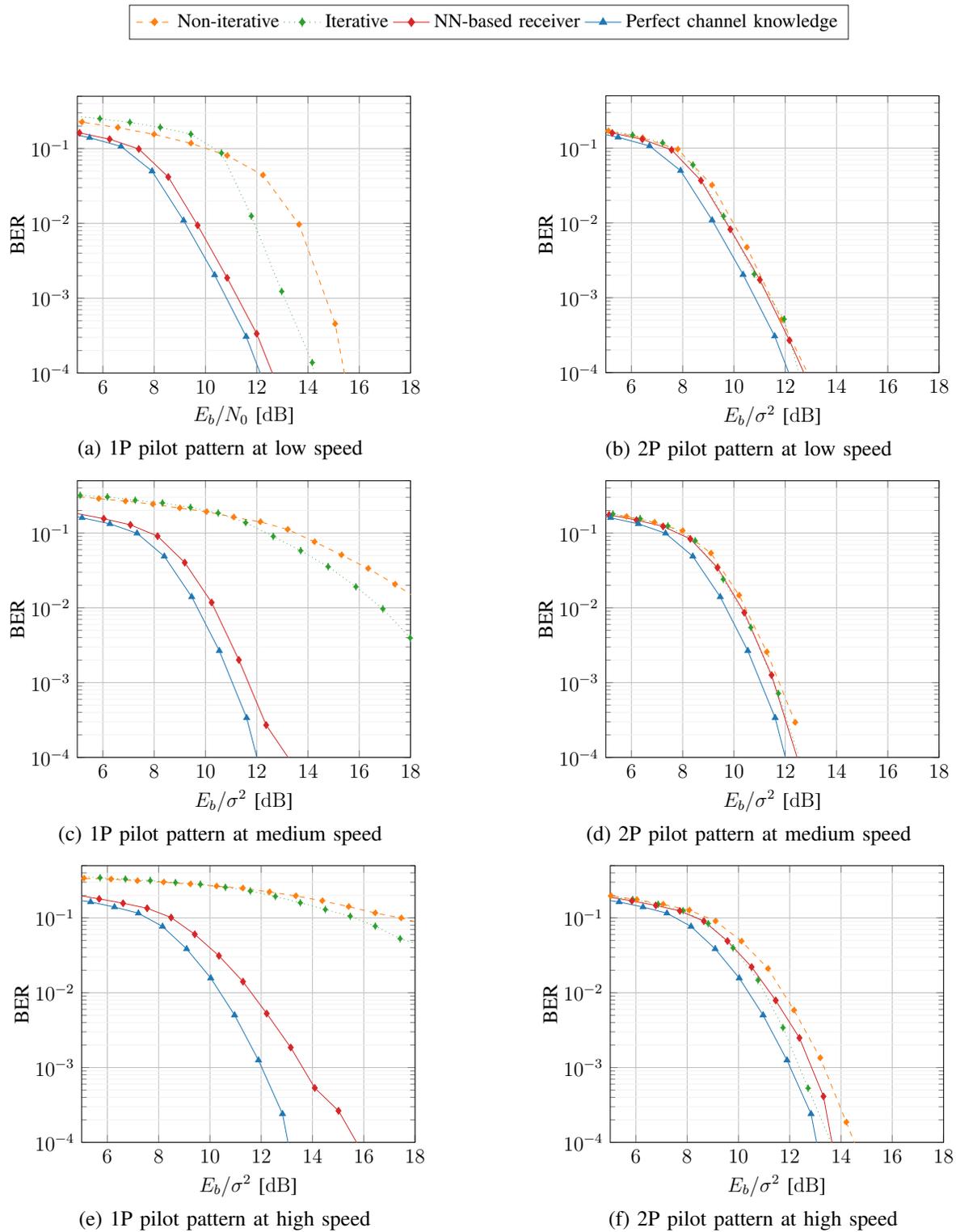
\begin{figure}
	\definecolor{C0}{HTML}{1F77B4}
	\definecolor{C1}{HTML}{FF7F0E}
	\definecolor{C2}{HTML}{2CA02C}
	\definecolor{C3}{HTML}{D62728}
	\definecolor{C4}{HTML}{9467BD}
 	\centering
 	\begin{subfigure}{1\linewidth}
 		\centering
		\begin{tikzpicture}
    		\begin{axis}[%
    			hide axis,
    			clip bounding box=upper bound,
    			legend style={draw=white!15!black,font=\footnotesize},
    			legend columns=4,
      			xmin=10,
    			xmax=10,
    			ymin=0,
    			ymax=0.1
    		]
    		\addlegendimage{C1,dashed, mark=diamond*}%
    		\addlegendentry{Non-iterative}%
   			\addlegendimage{C2,dotted, mark=diamond*}%
    		\addlegendentry{Iterative}%
   			\addlegendimage{C3,solid, mark=diamond*}%
    		\addlegendentry{NN-based receiver}%
    		\addlegendimage{C0,solid, mark=triangle*}%
    		\addlegendentry{Perfect channel knowledge}%
    		\end{axis}%
		\end{tikzpicture}
	\end{subfigure}
    \vspace{0.5\baselineskip}
	
 	\begin{subfigure}{0.45\linewidth}
 		\centering
		\begin{tikzpicture}[scale=0.8]
			\begin{semilogyaxis}[
				grid=both,
				grid style={line width=.4pt, draw=gray!10},
				major grid style={line width=.2pt,draw=gray!50},
				xlabel={$E_b/N_0$ [\si{\decibel}]},
				ylabel={BER},
				xmax=18,
				xmin=5,
				ymax=0.5,
				ymin=1e-4,
			]
				\addplot[C1,dashed, mark=diamond*] table [x=snr, y=ber, col sep=comma] {figs/qam_lmmse1p__ber_snr__speed_low.csv};		
				\addplot[C2,dotted, mark=diamond*] table [x=snr, y=ber, col sep=comma] {figs/qam_idd1p__ber_snr__speed_low.csv};		
				\addplot[C3,solid, mark=diamond*] table [x=snr, y=ber, col sep=comma] {figs/qam_nnrx1p__ber_snr__speed_low.csv};		
				\addplot[C0,solid, mark=triangle*] table [x=snr, y=ber, col sep=comma] {figs/qam_perfcsi__ber_snr__speed_low.csv};			
				\end{semilogyaxis}	
			\end{tikzpicture}
    		\vspace{-0.4\baselineskip}
			\caption{1P pilot pattern at low speed\label{fig:rx_ber_1p_splow_dsmed}}
	\end{subfigure}\quad\quad\quad
 	\begin{subfigure}{0.45\linewidth}
 		\centering
		\begin{tikzpicture}[scale=0.8]
			\begin{semilogyaxis}[
    			clip bounding box=upper bound,
				grid=both,
				grid style={line width=.4pt, draw=gray!10},
				major grid style={line width=.2pt,draw=gray!50},
				xlabel={$E_b/\sigma^2$ [\si{\decibel}]},
				ylabel={BER},
				xmax=18,
				xmin=5,
				ymax=0.5,
				ymin=1e-4,
			]
			
				\addplot[C1,dashed, mark=diamond*] table [x=snr, y=ber, col sep=comma] {figs/qam_lmmse2p__ber_snr__speed_low.csv};						
				\addplot[C2,dotted, mark=diamond*] table [x=snr, y=ber, col sep=comma] {figs/qam_idd2p__ber_snr__speed_low.csv};			
				\addplot[C3,solid, mark=diamond*] table [x=snr, y=ber, col sep=comma] {figs/qam_nnrx2p__ber_snr__speed_low.csv};		
				\addplot[C0,solid, mark=triangle*] table [x=snr, y=ber, col sep=comma] {figs/qam_perfcsi__ber_snr__speed_low.csv};			
				\end{semilogyaxis}	
			\end{tikzpicture}
    		\vspace{-0.4\baselineskip}
			\caption{2P pilot pattern at low speed\label{fig:rx_ber_2p_splow_dsmed}}
	\end{subfigure}
    \vspace{0.3\baselineskip}
	
 	\begin{subfigure}{0.45\linewidth}
 		\centering
		\begin{tikzpicture}[scale=0.8]
			\begin{semilogyaxis}[
    			clip bounding box=upper bound,
				grid=both,
				grid style={line width=.4pt, draw=gray!10},
				major grid style={line width=.2pt,draw=gray!50},
				xlabel={$E_b/\sigma^2$ [\si{\decibel}]},
				ylabel={BER},
				xmax=18,
				xmin=5,
				ymax=0.5,
				ymin=1e-4,
			]
				\addplot[C1,dashed, mark=diamond*] table [x=snr, y=ber, col sep=comma] {figs/qam_lmmse1p__ber_snr__speed_med.csv};		
				\addplot[C2,dotted, mark=diamond*] table [x=snr, y=ber, col sep=comma] {figs/qam_idd1p__ber_snr__speed_med.csv};		
				\addplot[C3,solid, mark=diamond*] table [x=snr, y=ber, col sep=comma] {figs/qam_nnrx1p__ber_snr__speed_med.csv};		
				\addplot[C0,solid, mark=triangle*] table [x=snr, y=ber, col sep=comma] {figs/qam_perfcsi__ber_snr__speed_med.csv};			
				\end{semilogyaxis}	
			\end{tikzpicture}
    		\vspace{-0.4\baselineskip}
			\caption{1P pilot pattern at medium speed\label{fig:rx_ber_1p_spmed_dsmed}}
	\end{subfigure}\quad\quad\quad
 	\begin{subfigure}{0.45\linewidth}
 		\centering
		\begin{tikzpicture}[scale=0.8]
			\begin{semilogyaxis}[
    			clip bounding box=upper bound,
				grid=both,
				grid style={line width=.4pt, draw=gray!10},
				major grid style={line width=.2pt,draw=gray!50},
				xlabel={$E_b/\sigma^2$ [\si{\decibel}]},
				ylabel={BER},
				xmax=18,
				xmin=5,
				ymax=0.5,
				ymin=1e-4,
			]
				\addplot[C1,dashed, mark=diamond*] table [x=snr, y=ber, col sep=comma] {figs/qam_lmmse2p__ber_snr__speed_med.csv};		
				\addplot[C2,dotted, mark=diamond*] table [x=snr, y=ber, col sep=comma] {figs/qam_idd2p__ber_snr__speed_med.csv};			
				\addplot[C3,solid, mark=diamond*] table [x=snr, y=ber, col sep=comma] {figs/qam_nnrx2p__ber_snr__speed_med.csv};		
				\addplot[C0,solid, mark=triangle*] table [x=snr, y=ber, col sep=comma] {figs/qam_perfcsi__ber_snr__speed_med.csv};			
				\end{semilogyaxis}	
			\end{tikzpicture}
    		\vspace{-0.4\baselineskip}
			\caption{2P pilot pattern at medium speed\label{fig:rx_ber_2p_spmed_dsmed}}
	\end{subfigure}
    \vspace{0.3\baselineskip}
	
 	\begin{subfigure}{0.45\linewidth}
 		\centering
		\begin{tikzpicture}[scale=0.8]
			\begin{semilogyaxis}[
    			clip bounding box=upper bound,
				grid=both,
				grid style={line width=.4pt, draw=gray!10},
				major grid style={line width=.2pt,draw=gray!50},
				xlabel={$E_b/\sigma^2$ [\si{\decibel}]},
				ylabel={BER},
				xmax=18,
				xmin=5,
				ymax=0.5,
				ymin=1e-4,
			]
				\addplot[C1,dashed, mark=diamond*] table [x=snr, y=ber, col sep=comma] {figs/qam_lmmse1p__ber_snr__speed_high.csv};		
				\addplot[C2,dotted, mark=diamond*] table [x=snr, y=ber, col sep=comma] {figs/qam_idd1p__ber_snr__speed_high.csv};		
				\addplot[C3,solid, mark=diamond*] table [x=snr, y=ber, col sep=comma] {figs/qam_nnrx1p__ber_snr__speed_high.csv};		
				\addplot[C0,solid, mark=triangle*] table [x=snr, y=ber, col sep=comma] {figs/qam_perfcsi__ber_snr__speed_high.csv};			
				\end{semilogyaxis}	
			\end{tikzpicture}
    		\vspace{-0.4\baselineskip}
			\caption{1P pilot pattern at high speed\label{fig:rx_ber_1p_sphigh_dsmed}}
	\end{subfigure}\quad\quad\quad
 	\begin{subfigure}{0.45\linewidth}
 		\centering
		\begin{tikzpicture}[scale=0.8]
			\begin{semilogyaxis}[
    			clip bounding box=upper bound,
				grid=both,
				grid style={line width=.4pt, draw=gray!10},
				major grid style={line width=.2pt,draw=gray!50},
				xlabel={$E_b/\sigma^2$ [\si{\decibel}]},
				ylabel={BER},
				xmax=18,
				xmin=5,
				ymax=0.5,
				ymin=1e-4,
			]		
				\addplot[C1,dashed, mark=diamond*] table [x=snr, y=ber, col sep=comma] {figs/qam_lmmse2p__ber_snr__speed_high.csv};		
				\addplot[C2,dotted, mark=diamond*] table [x=snr, y=ber, col sep=comma] {figs/qam_idd2p__ber_snr__speed_high.csv};		
				\addplot[C3,solid, mark=diamond*] table [x=snr, y=ber, col sep=comma] {figs/qam_nnrx2p__ber_snr__speed_high.csv};		
				\addplot[C0,solid, mark=triangle*] table [x=snr, y=ber, col sep=comma] {figs/qam_perfcsi__ber_snr__speed_high.csv};			

				\end{semilogyaxis}	
			\end{tikzpicture}
    		\vspace{-0.4\baselineskip}
			\caption{2P pilot pattern at high speed\label{fig:rx_ber_2p_sphigh_dsmed}}
	\end{subfigure}
	\caption{BER achieved by the evaluated receivers for different ranges of speed\label{fig:rx_ber}}
\end{figure}

In this section, we will focus only on a learned receiver as described above using conventional \gls{QAM} and pilot patterns from \gls{5GNR} (see Fig.~\ref{fig:pilot}).
An \gls{NN}-based receiver was optimized for each pilot pattern.
By training only the receiver on the total \gls{BCE}~(\ref{eq:bce}), one minimizes the \gls{KL} divergence between the posterior distribution on the transmitted bits implemented by the \gls{NN}-based receiver and the true posterior distribution, as the first term in~(\ref{eq:rate}) depends only on the transmitter and channel.
In other words, the receiver is optimized to compute values close to the true posterior distribution on the transmitted bits given the channel output.

The first column of Fig.~\ref{fig:rx_ber} shows the \glspl{BER} achieved by the evaluated schemes with the pilot pattern 1P for the three speed ranges.
As one can see, only the \gls{NN}-based receiver achieves \glspl{BER} within \SI{2}{\decibel} from the perfect channel knowledge bound for all speed ranges.
At low speeds (Fig.~\ref{fig:rx_ber_1p_splow_dsmed}), the \gls{NN}-based receiver enables \glspl{BER} less than \SI{0.5}{\decibel} from the perfect channel knowledge bound.
All schemes experience higher \glspl{BER} as the speed increases (Fig~\ref{fig:rx_ber_1p_spmed_dsmed} and Fig.~\ref{fig:rx_ber_1p_sphigh_dsmed}).
However, conventional baselines are less robust to higher speeds compared to the neural receiver.
This is especially true for the iterative scheme, which enables significant gains at low speeds compared to the non-iterative approach, but which vanish at high speeds.

Similarly, the second column of Fig.~\ref{fig:rx_ber} compares the \glspl{BER} of all schemes using the 2P pilot patter.
As one can see, all schemes have \glspl{BER} close to the perfect channel knowledge limit. 
At high speeds (Fig.~\ref{fig:rx_ber_2p_sphigh_dsmed}), the iterative baseline achieves the lowest \glspl{BER}, closely followed by the neural receiver which outperforms the non-iterative baseline.
Similar simulations where conducted considering three delay spread ranges.
The results revealed no significant impact of the delay spread on the relative performance, and are therefore not shown in this paper.

\begin{figure}
 	\centering

 	\begin{subfigure}{\linewidth}
 		\centering
		\includegraphics[scale=0.5]{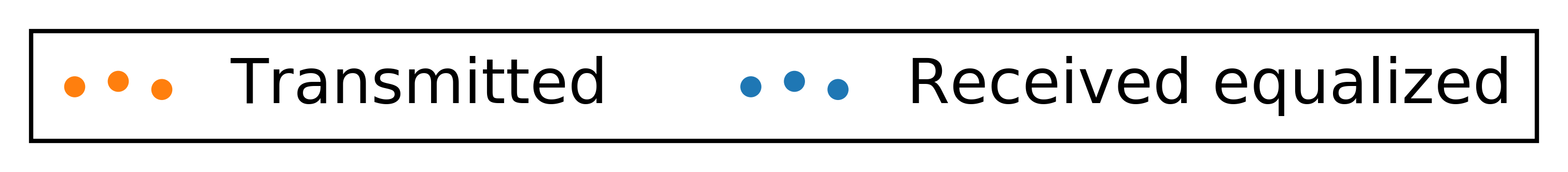}
	\end{subfigure} 	
 	
 	\begin{subfigure}{\linewidth}
 		\centering
		\includegraphics[scale=0.55]{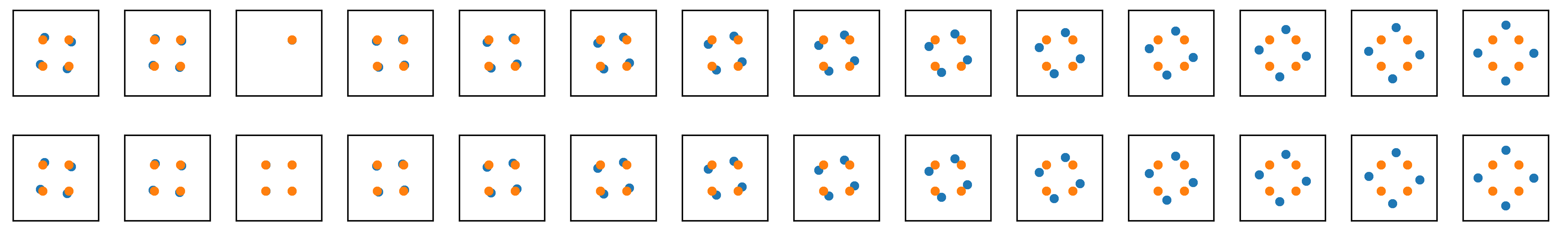}
		\caption{Constellations with pilot pattern ``1P''\label{fig:consts_p1}}
	\end{subfigure}
	
 	\begin{subfigure}{\linewidth}
 		\centering
		\includegraphics[scale=0.55]{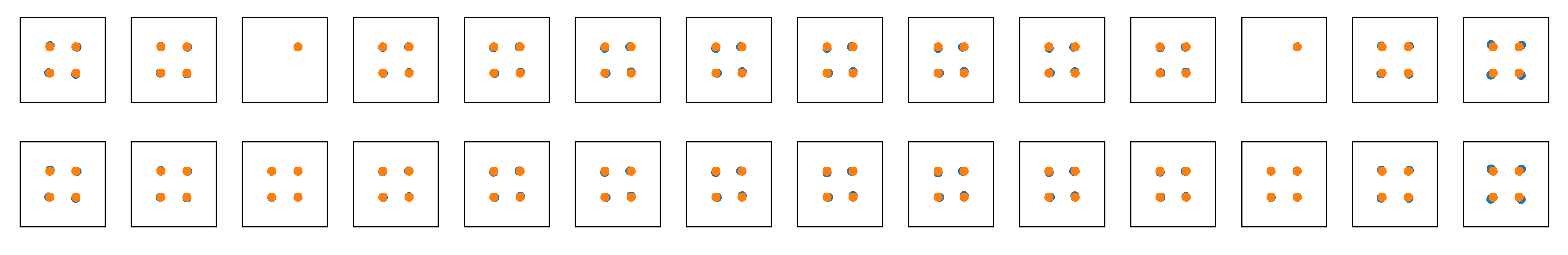}
		\caption{Constellations with pilot pattern ``2P''\label{fig:consts_p2}}
	\end{subfigure}
	\caption{Equalized received symbols at high speed, medium delay spread, and infinite SNR. The effect of a single channel realization is shown. The \gls{RE} containing a single point corresponds to the pilots.\label{fig:consts}}
\end{figure}

To get insight into why conventional baselines fail when using the 1P pilot pattern, Fig.~\ref{fig:consts} shows the equalized received symbols using the non-iterative baseline for one slot and 2 subcarriers in the high speed scenario (i.e., corresponding to the \glspl{BER} of Fig.~\ref{fig:rx_ber_1p_sphigh_dsmed} and Fig.~\ref{fig:rx_ber_2p_sphigh_dsmed}).
This figure was generated considering the effect of a single channel realization, and assuming an infinite \gls{SNR}, such that only the effect of channel aging appears.
A \gls{QPSK} constellation was used for readability.
The \glspl{RE} in which a single point is visible correspond to pilots.
As one can see, when the 1P pilot pattern is used (Fig.~\ref{fig:consts_p1}), equalized symbols match the ground-truth constellation for \glspl{RE} located near the pilots.
However, for \glspl{RE} located further away from the reference signals, i.e., located at the end of the slot, channel aging causes significant mismatch between the true channel coefficients and the estimated ones, despite the infinite \gls{SNR}.
This effect is amplified by the channel noise, and leads to the high error rates observed when using such a pilot pattern.
One can therefore suppose that the \gls{NN}-based receiver is able to better correct the effects of channel aging compared to conventional baselines, by possibly leveraging the data-carrying symbols for improved channel estimation.
When using the 2P pilot pattern (Fig.~\ref{fig:consts_p2}), the presence of a second pilot at the end of the slot reduces the effect of channel aging, leading to much lower error rates.

From these results, one can conclude that \gls{NN}-based receivers are mostly beneficial in high-speed scenarios and/or for sparse pilot patterns.
We have observed that we can drastically reduce the number of pilots, in some cases down to a single subcarrier carrying two pilots, without significant increase of \gls{BER}.
A similar observation was made in~\cite{8052521} regarding the robustness to a lower number of pilots.
Although lower \glspl{BER} can be achieved by transmitting more pilots with conventional baselines, this comes at the cost of lower throughput.
Therefore, an \gls{NN}-based receiver can enable higher throughputs.
The next section demonstrates that joint optimization of the transmitter and receiver enables further throughput gains by completely removing the need for orthogonal reference signals.

\subsection{Evaluation of end-to-end learning}
\label{sec:e2e_res}

The two end-to-end schemes presented in Sections~\ref{sec:e2e_gs} and~\ref{sec:e2e_sip} are evaluated in this section.
In addition to the previously introduced baselines, joint optimization of the constellation geometry and bit labeling, assuming perfect channel knowledge at the receiver is also considered.

\begin{figure}
 	\centering
	\includegraphics[scale=0.45]{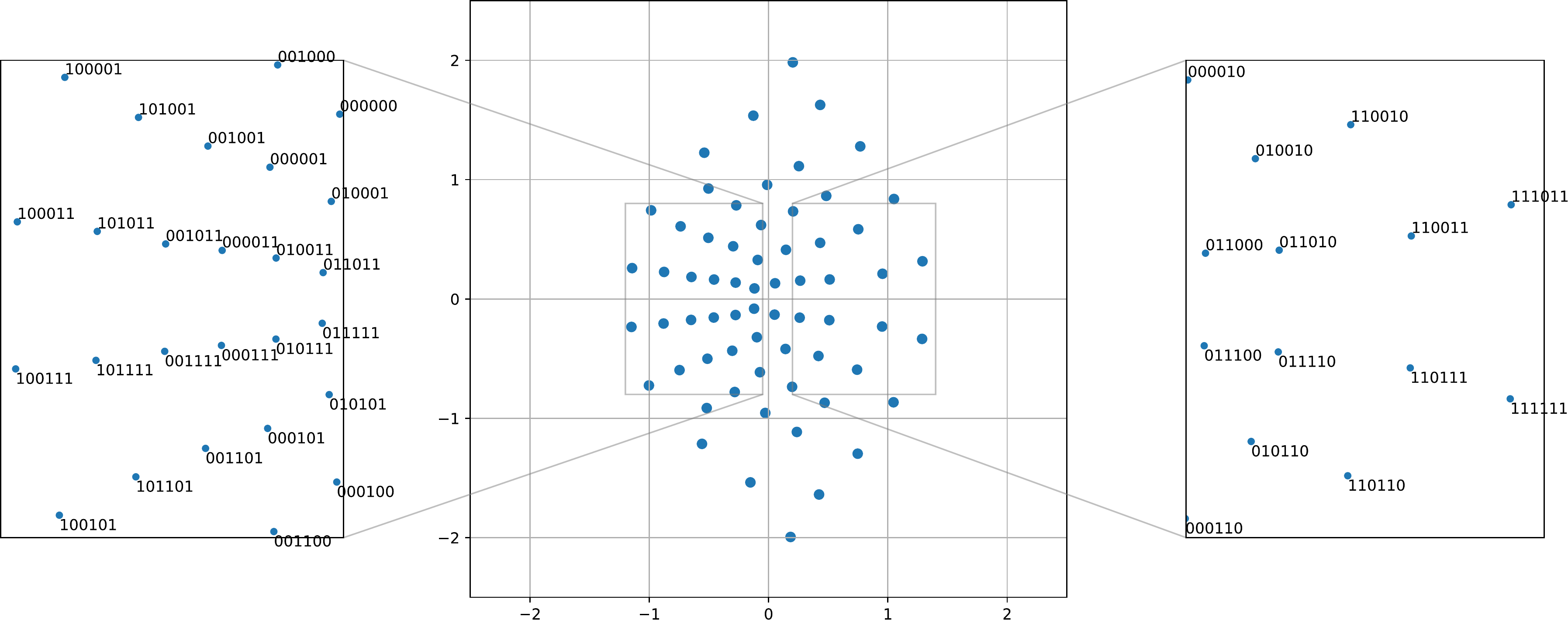}
	\caption{Constellation obtained from end-to-end learning which eliminates the need for orthogonal or superimposed pilots\label{fig:gs_const}}
\end{figure}

Fig.~\ref{fig:gs_const} shows the constellation and labeling obtained by training the end-to-end system introduced in Section~\ref{sec:e2e_gs}, and referred to as the \gls{GS} scheme.
One can see that the learned constellation has a unique horizontal axis of symmetry.
Moreover, a form of Gray labeling was learned, such that points next to each other differ by one bit.

\begin{figure}
 	\centering
	\includegraphics[scale=0.45]{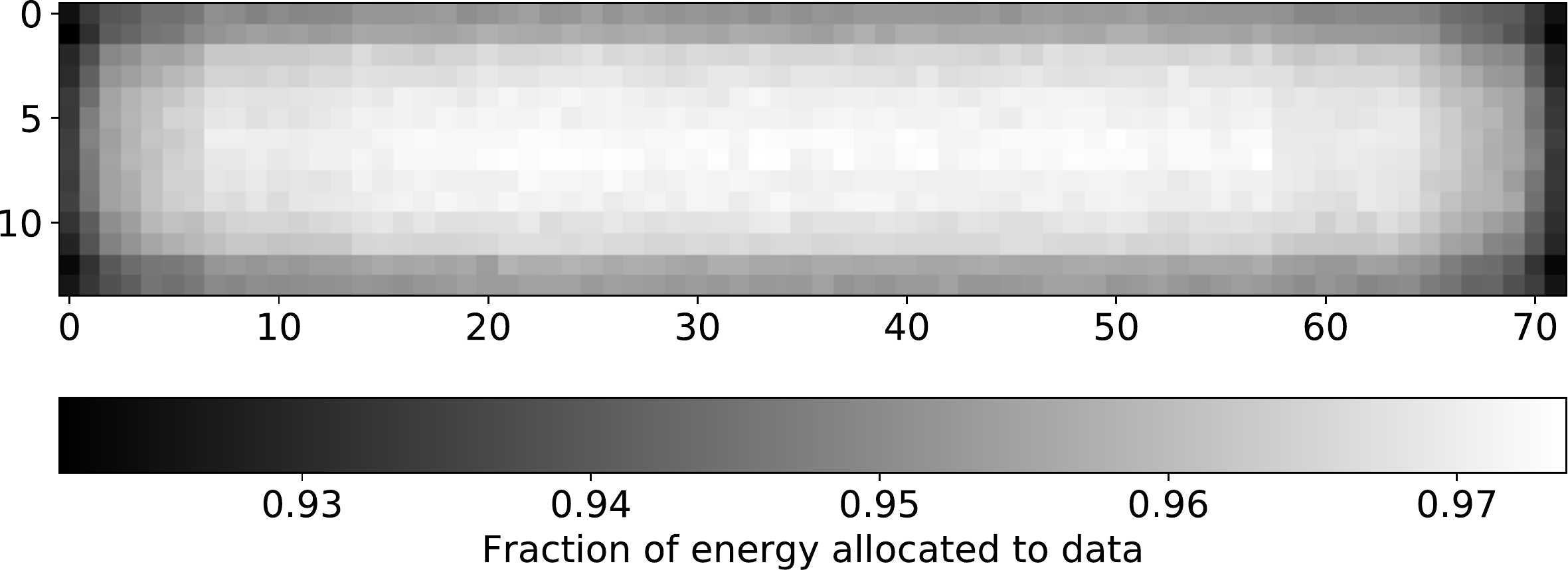}
	\caption{SIP allocation learned by the end-to-end system\label{fig:sip_map}}
\end{figure}

Regarding the \gls{QAM}-\gls{SIP} scheme introduced in Section~\ref{sec:e2e_sip}, the pilot allocation obtained by optimizing the system is shown in Fig.~\ref{fig:sip_map}.
One can see that for all \glspl{RE}, less than $10\%$ of the available energy is allocated to \glspl{SIP}.
The \glspl{RE} located at the edges of the frame have a higher ratio of their energy allocated to \glspl{SIP} compared to the other \glspl{RE}, especially, the first and last subcarriers of the frame.
However, overall, the variation across the frame is rather small.
Thus, optimizing a single energy-level used by all the \glspl{RE} could be sufficient and reduce implementation complexity.

\begin{figure}
	\definecolor{C0}{HTML}{1F77B4}
	\definecolor{C1}{HTML}{FF7F0E}
	\definecolor{C2}{HTML}{2CA02C}
	\definecolor{C3}{HTML}{D62728}
	\definecolor{C4}{HTML}{9467BD}
	\definecolor{C5}{HTML}{8c564b}
	\definecolor{C6}{HTML}{e377c2}
	\definecolor{C7}{HTML}{7f7f7f}
 	\centering
 	\begin{subfigure}{1\linewidth}
 		\centering
		\begin{tikzpicture}
    		\begin{axis}[%
    			hide axis,
    			clip bounding box=upper bound,
    			legend style={draw=white!15!black,font=\footnotesize},
    			legend columns=3,
      			xmin=10,
    			xmax=10,
    			ymin=0,
    			ymax=0.1
    		]
    		\addlegendimage{C2,dashed, mark=diamond*}%
    		\addlegendentry{QAM \& Iterative 1P}%
   			\addlegendimage{C3,solid, mark=diamond*}%
    		\addlegendentry{QAM \& NN-based receiver 1P}%
    		\addlegendimage{C0,solid, mark=triangle*}%
    		\addlegendentry{QAM \& Perfect channel knowledge}%
        	\addlegendimage{C2,dashed, mark=square*}%
    		\addlegendentry{QAM \& Iterative 2P}%
    		\addlegendimage{C3,solid, mark=square*}%
    		\addlegendentry{QAM \& NN-based receiver 2P}%
    		\addlegendimage{C4,solid, mark=asterisk}%
    		\addlegendentry{GS \& Perfect channel knowledge}%
    		\addlegendimage{C5,solid, mark=*}%
       		\addlegendentry{GS \& NN-based receiver}%
     		\addlegendimage{C6,solid, mark=pentagon*}%
    		\addlegendentry{QAM-SIP \& NN-based receiver}%

    		\end{axis}%
		\end{tikzpicture}
	\end{subfigure}
    \vspace{0.5\baselineskip}
	
 	\begin{subfigure}{0.45\linewidth}
 		\centering
		\begin{tikzpicture}[scale=0.8]
			\begin{semilogyaxis}[
				grid=both,
				grid style={line width=.4pt, draw=gray!10},
				major grid style={line width=.2pt,draw=gray!50},
				xlabel={$E_b/\sigma^2$ [\si{\decibel}]},
				ylabel={BER},
				xmax=18,
				xmin=5,
				ymax=0.5,
				ymin=1e-4,
			]
				\addplot[C3,solid, mark=diamond*] table [x=snr, y=ber, col sep=comma] {figs/qam_nnrx1p__ber_snr__speed_low.csv};
				\addplot[C5,solid, mark=*] table [x=snr, y=ber, col sep=comma] {figs/gs_nnrx__ber_snr__speed_low.csv};		
				\addplot[C0,solid, mark=triangle*] table [x=snr, y=ber, col sep=comma] {figs/qam_perfcsi__ber_snr__speed_low.csv};		
				\addplot[C3,solid, mark=square*] table [x=snr, y=ber, col sep=comma] {figs/qam_nnrx2p__ber_snr__speed_low.csv};		
				\addplot[C6,solid, mark=pentagon*] table [x=snr, y=ber, col sep=comma] {figs/qam_nnrxsip__ber_snr__speed_low.csv};		
				\addplot[C4,solid, mark=asterisk] table [x=snr, y=ber, col sep=comma] {figs/gs_perfcsi__ber_snr__speed_low.csv};	
				\end{semilogyaxis}	
			\end{tikzpicture}
    		\vspace{-0.4\baselineskip}
			\caption{BER at low speed\label{fig:ber_speed_low}}
	\end{subfigure}\quad\quad\quad
 	\begin{subfigure}{0.45\linewidth}
 		\centering
		\begin{tikzpicture}[scale=0.8]
			\begin{axis}[
    			clip bounding box=upper bound,
				grid=both,
				grid style={line width=.4pt, draw=gray!10},
				major grid style={line width=.2pt,draw=gray!50},
				xlabel={$E_s/\sigma^2$ [\si{\decibel}]},
				ylabel={Goodput [\si{\bit\per\frame}]},
				xmax=25,
				xmin=5,
			]
				\addplot[C2,dashed, mark=diamond*] table [x=snr, y=goodput, col sep=comma] {figs/qam_idd1p__goodput_snr__speed_low.csv};
				\addplot[C3,solid, mark=diamond*] table [x=snr, y=goodput, col sep=comma] {figs/qam_nnrx1p__goodput_snr__speed_low.csv};
				\addplot[C5,solid, mark=*] table [x=snr, y=goodput, col sep=comma] {figs/gs_nnrx__goodput_snr__speed_low.csv};		
				\addplot[C0,solid, mark=triangle*] table [x=snr, y=goodput, col sep=comma] {figs/qam_perfcsi__goodput_snr__speed_low.csv};		
				\addplot[C2,dashed, mark=square*] table [x=snr, y=goodput, col sep=comma] {figs/qam_idd2p__goodput_snr__speed_low.csv};
				\addplot[C3,solid, mark=square*] table [x=snr, y=goodput, col sep=comma] {figs/qam_nnrx2p__goodput_snr__speed_low.csv};	
				\addplot[C6,solid, mark=pentagon*] table [x=snr, y=goodput, col sep=comma] {figs/qam_nnrxsip__goodput_snr__speed_low.csv};		
				\addplot[C4,solid, mark=asterisk] table [x=snr, y=goodput, col sep=comma] {figs/gs_perfcsi__goodput_snr__speed_low.csv};	
				\end{axis}	
			\end{tikzpicture}
    		\vspace{-0.4\baselineskip}
			\caption{Goodput at low speed\label{fig:gp_speed_low}}
	\end{subfigure}
    \vspace{0.3\baselineskip}
	
 	\begin{subfigure}{0.45\linewidth}
 		\centering
		\begin{tikzpicture}[scale=0.8]
			\begin{semilogyaxis}[
    			clip bounding box=upper bound,
				grid=both,
				grid style={line width=.4pt, draw=gray!10},
				major grid style={line width=.2pt,draw=gray!50},
				xlabel={$E_b/\sigma^2$ [\si{\decibel}]},
				ylabel={BER},
				xmax=18,
				xmin=5,
				ymax=0.5,
				ymin=1e-4,
			]
				\addplot[C3,solid, mark=diamond*] table [x=snr, y=ber, col sep=comma] {figs/qam_nnrx1p__ber_snr__speed_med.csv};		
				\addplot[C5,solid, mark=*] table [x=snr, y=ber, col sep=comma] {figs/gs_nnrx__ber_snr__speed_med.csv};		
				\addplot[C0,solid, mark=triangle*] table [x=snr, y=ber, col sep=comma] {figs/qam_perfcsi__ber_snr__speed_med.csv};		
				\addplot[C3,solid, mark=square*] table [x=snr, y=ber, col sep=comma] {figs/qam_nnrx2p__ber_snr__speed_med.csv};		
				\addplot[C6,solid, mark=pentagon*] table [x=snr, y=ber, col sep=comma] {figs/qam_nnrxsip__ber_snr__speed_med.csv};		
				\addplot[C4,solid, mark=asterisk] table [x=snr, y=ber, col sep=comma] {figs/gs_perfcsi__ber_snr__speed_med.csv};	
				\end{semilogyaxis}	
			\end{tikzpicture}
    		\vspace{-0.4\baselineskip}
			\caption{BER at medium speed\label{fig:ber_speed_med}}
	\end{subfigure}\quad\quad\quad
 	\begin{subfigure}{0.45\linewidth}
 		\centering
		\begin{tikzpicture}[scale=0.8]
			\begin{axis}[
    			clip bounding box=upper bound,
				grid=both,
				grid style={line width=.4pt, draw=gray!10},
				major grid style={line width=.2pt,draw=gray!50},
				xlabel={$E_s/\sigma^2$ [\si{\decibel}]},
				ylabel={Goodput [\si{\bit\per\frame}]},
				xmax=25,
				xmin=5,
			]
				\addplot[C2,dashed, mark=diamond*] table [x=snr, y=goodput, col sep=comma] {figs/qam_idd1p__goodput_snr__speed_med.csv};	
				\addplot[C3,solid, mark=diamond*] table [x=snr, y=goodput, col sep=comma] {figs/qam_nnrx1p__goodput_snr__speed_med.csv};	
				\addplot[C5,solid, mark=*] table [x=snr, y=goodput, col sep=comma] {figs/gs_nnrx__goodput_snr__speed_med.csv};		
				\addplot[C0,solid, mark=triangle*] table [x=snr, y=goodput, col sep=comma] {figs/qam_perfcsi__goodput_snr__speed_med.csv};		
				\addplot[C2,dashed, mark=square*] table [x=snr, y=goodput, col sep=comma] {figs/qam_idd2p__goodput_snr__speed_med.csv};
				\addplot[C3,solid, mark=square*] table [x=snr, y=goodput, col sep=comma] {figs/qam_nnrx2p__goodput_snr__speed_med.csv};	
				\addplot[C6,solid, mark=pentagon*] table [x=snr, y=goodput, col sep=comma] {figs/qam_nnrxsip__goodput_snr__speed_med.csv};		
				\addplot[C4,solid, mark=asterisk] table [x=snr, y=goodput, col sep=comma] {figs/gs_perfcsi__goodput_snr__speed_med.csv};		
				\end{axis}	
			\end{tikzpicture}
    		\vspace{-0.4\baselineskip}
			\caption{Goodput at medium speed\label{fig:gp_speed_med}}
	\end{subfigure}
    \vspace{0.3\baselineskip}
	
 	\begin{subfigure}{0.45\linewidth}
 		\centering
		\begin{tikzpicture}[scale=0.8]
			\begin{semilogyaxis}[
    			clip bounding box=upper bound,
				grid=both,
				grid style={line width=.4pt, draw=gray!10},
				major grid style={line width=.2pt,draw=gray!50},
				xlabel={$E_b/\sigma^2$ [\si{\decibel}]},
				ylabel={BER},
				xmax=18,
				xmin=5,
				ymax=0.5,
				ymin=1e-4,
			]
				\addplot[C3,solid, mark=diamond*] table [x=snr, y=ber, col sep=comma] {figs/qam_nnrx1p__ber_snr__speed_high.csv};		
				\addplot[C5,solid, mark=*] table [x=snr, y=ber, col sep=comma] {figs/gs_nnrx__ber_snr__speed_high.csv};		
				\addplot[C0,solid, mark=triangle*] table [x=snr, y=ber, col sep=comma] {figs/qam_perfcsi__ber_snr__speed_high.csv};		
				\addplot[C3,solid, mark=square*] table [x=snr, y=ber, col sep=comma] {figs/qam_nnrx2p__ber_snr__speed_high.csv};		
				\addplot[C6,solid, mark=pentagon*] table [x=snr, y=ber, col sep=comma] {figs/qam_nnrxsip__ber_snr__speed_high.csv};		
				\addplot[C4,solid, mark=asterisk] table [x=snr, y=ber, col sep=comma] {figs/gs_perfcsi__ber_snr__speed_high.csv};	
				\end{semilogyaxis}	
			\end{tikzpicture}
    		\vspace{-0.4\baselineskip}
			\caption{BER at high speed\label{fig:ber_speed_high}}
	\end{subfigure}\quad\quad\quad
 	\begin{subfigure}{0.45\linewidth}
 		\centering
		\begin{tikzpicture}[scale=0.8]
			\begin{axis}[
    			clip bounding box=upper bound,
				grid=both,
				grid style={line width=.4pt, draw=gray!10},
				major grid style={line width=.2pt,draw=gray!50},
				xlabel={$E_s/\sigma^2$ [\si{\decibel}]},
				ylabel={Goodput [\si{\bit\per\frame}]},
				xmax=25,
				xmin=5,
			]		
				\addplot[C2,dashed, mark=diamond*] table [x=snr, y=goodput, col sep=comma] {figs/qam_idd1p__goodput_snr__speed_high.csv};
				\addplot[C3,solid, mark=diamond*] table [x=snr, y=goodput, col sep=comma] {figs/qam_nnrx1p__goodput_snr__speed_high.csv};		
				\addplot[C5,solid, mark=*] table [x=snr, y=goodput, col sep=comma] {figs/gs_nnrx__goodput_snr__speed_high.csv};		
				\addplot[C0,solid, mark=triangle*] table [x=snr, y=goodput, col sep=comma] {figs/qam_perfcsi__goodput_snr__speed_high.csv};		
				\addplot[C2,dashed, mark=square*] table [x=snr, y=goodput, col sep=comma] {figs/qam_idd2p__goodput_snr__speed_high.csv};
				\addplot[C3,solid, mark=square*] table [x=snr, y=goodput, col sep=comma] {figs/qam_nnrx2p__goodput_snr__speed_high.csv};		
				\addplot[C6,solid, mark=pentagon*] table [x=snr, y=goodput, col sep=comma] {figs/qam_nnrxsip__goodput_snr__speed_high.csv};		
				\addplot[C4,solid, mark=asterisk] table [x=snr, y=goodput, col sep=comma] {figs/gs_perfcsi__goodput_snr__speed_high.csv};	
				\end{axis}	
			\end{tikzpicture}
    		\vspace{-0.4\baselineskip}
			\caption{Goodput at high speed\label{fig:gp_speed_high}}
	\end{subfigure}
	\caption{BER and goodput achieved by the evaluated schemes for different ranges of speed\label{fig:gp_ber}}
\end{figure}
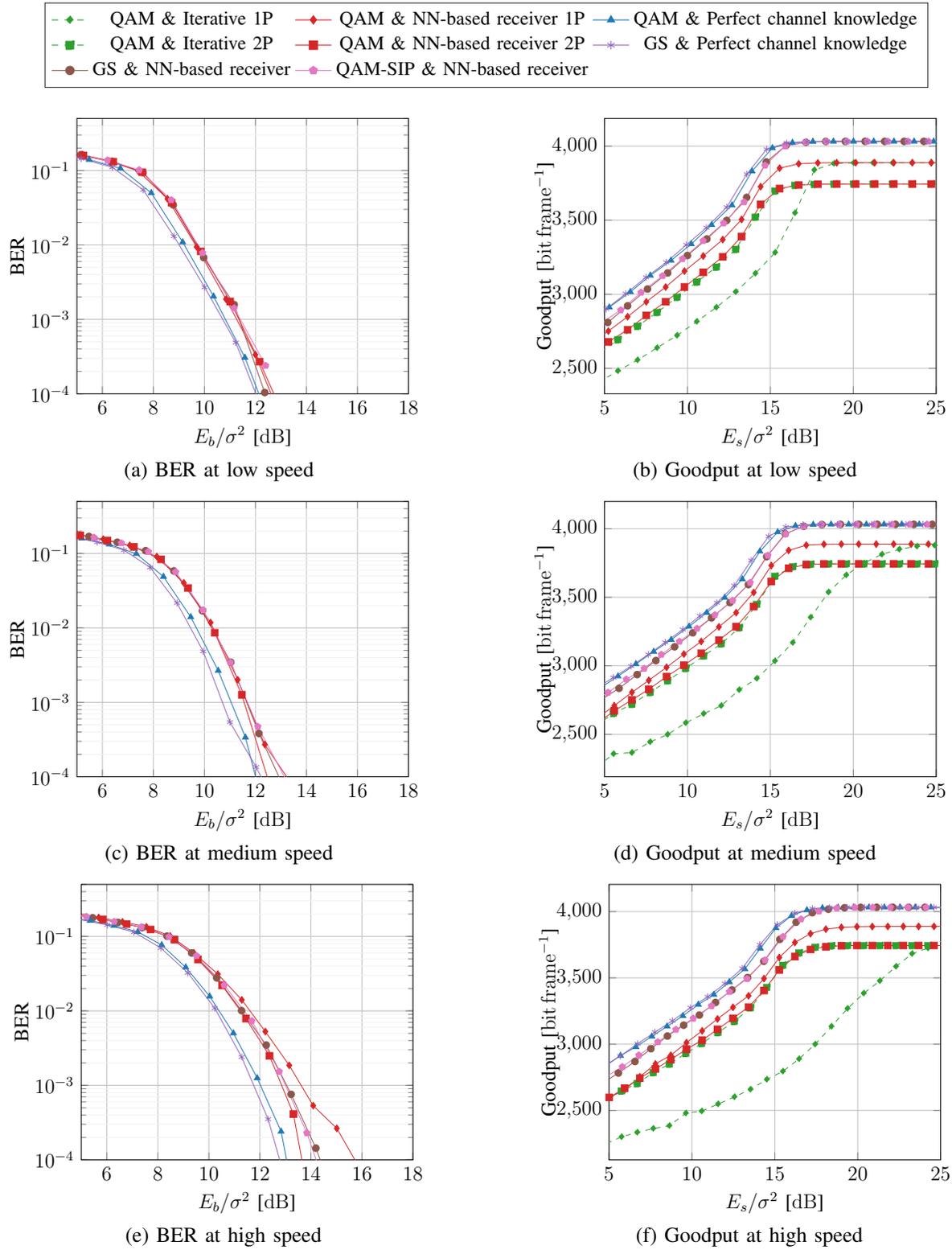

The first column of Fig.~\ref{fig:gp_ber} shows the \glspl{BER} achieved by the various schemes.
For readability, \gls{QAM} with the non-iterative (Section~\ref{sec:non_it}) and iterative (Section~\ref{sec:iedd}) receivers are not shown, as these schemes achieve similar or higher \glspl{BER} than \gls{QAM} with the \gls{NN}-based receiver (Section~\ref{sec:rx_eval}).
As one can see, when perfect channel knowledge at the receiver is assumed, optimization of the constellation geometry and labeling only makes little difference.
The \gls{GS} and \gls{QAM}-\gls{SIP} schemes  achieve \glspl{BER} similar to the ones of \gls{QAM} with orthogonal pilots and the \gls{NN}-based receiver.
This is the case regardless of which orthogonal pilot pattern is used.
Moreover, these results hold for the three considered speed ranges.

The benefits of achieving low \glspl{BER} without the requirement of transmitting orthogonal pilots, as allowed by the \gls{GS} and \gls{QAM}-\gls{SIP} schemes, is that it enables higher \emph{goodput}, as shown in the second column of Fig.~\ref{fig:gp_ber}.
The non-iterative receiver baseline was omitted as it achieves similar or worse performance than the iterative one.
The goodput measures the number of bits per frame successfully received, and is defined by
\begin{equation}
\text{Goodput} \coloneqq r \rho m n \LB 1 - \text{BER}\RB
\end{equation}
where $n$ is the number of \glspl{RE} forming a frame, and $\rho$ is the ratio of data carrying \glspl{RE} within a \gls{PRB} ($\rho = 1$ for \gls{GS}, \gls{QAM}-\gls{SIP}, and when perfect channel knowledge is assumed at the receiver, $\rho = \frac{162}{168}$ when the orthogonal 1P pilot pattern is used, and $\rho = \frac{156}{168}$ when the orthogonal 2P pilot pattern is used).
Moreover, because the goodput accounts for the unequal average number of information bits transmitted per \gls{RE} among the different schemes through the parameter $\rho$, it is plotted with respect to the energy per symbol to noise power spectral density ratio, defined as
\begin{equation}
\frac{E_s}{\sigma^2} \coloneqq \frac{\sum_{i=1}^{n_S} \sum_{k=1}^{n_T} \abs{H_{i,k}}^2}{n\sigma^2}.
\end{equation}
These plots show that an \gls{NN}-based receiver with \gls{QAM} and orthogonal pilots enables close to \SI{25}{\percent} higher goodput than the receiver baselines when the 1P pilot pattern is used at high speeds (Fig.~\ref{fig:gp_speed_high}).
One can see that for all the considered speed ranges, the \gls{GS} and \gls{QAM}-\gls{SIP} schemes achieve goodput close to the one with perfect channel knowledge at the receiver, especially when $E_s/\sigma^2$ is higher than \SI{15}{\decibel}.
The orthogonal pilot-based schemes saturate at lower values as some of the \glspl{RE} are allocated to reference signals.
The additional gains enabled by end-to-end learning range from \SI{4}{\percent} to \SI{8}{\percent} depending on which pilot pattern is used by the baselines.
Therefore, one can conclude from these results that the majority of the gains enabled by end-to-end learning can be achieved by leveraging an \gls{NN}-based receiver with sparse pilot patterns.
However, joint optimization of the transmitter and receiver is required if one wants to achieve the highest possible gains by suppressing all orthogonal pilots. 

\subsection{PAPR study}

\begin{figure}
	\definecolor{C0}{HTML}{1F77B4}
	\definecolor{C1}{HTML}{FF7F0E}
	\definecolor{C2}{HTML}{2CA02C}
	\definecolor{C3}{HTML}{D62728}
	\definecolor{C4}{HTML}{9467BD}
 	\centering	
		\begin{tikzpicture}
			\begin{axis}[
				scale=0.8,
				grid=both,
				grid style={line width=.4pt, draw=gray!10},
				major grid style={line width=.2pt,draw=gray!50},
				xlabel={PAPR [\si{\decibel}]},
				ylabel={CDF},
     			legend style={draw=white!15!black,font=\footnotesize,at={(0.4,0.9)}},
     			every axis plot/.append style={ultra thick}
			]
				\addplot[C0,solid] table [x=PAR, y=CDF, col sep=comma] {figs/papr_qam.csv};
    			\addlegendentry{QAM}%

				\addplot[C1,dashed] table [x=PAR, y=CDF, col sep=comma] {figs/papr_qamsip.csv};
    			\addlegendentry{QAM-SIP}%
	
				\addplot[C2,dotted] table [x=PAR, y=CDF, col sep=comma] {figs/papr_gs.csv};	
    			\addlegendentry{GS}%
	
				\end{axis}	
			\end{tikzpicture}

	\caption{CDF of the \gls{PAPR} \label{fig:papr}}
\end{figure}
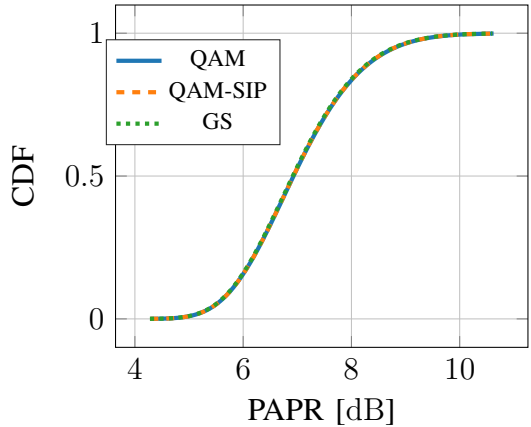

We conclude this section by evaluating the \gls{PAPR} incurred by the evaluated approaches.
Fig.~\ref{fig:papr} shows the \gls{CDF} of the \gls{PAPR} of \gls{QAM}, \gls{QAM}-\gls{SIP}, and \gls{GS}.
For each scheme, $7\times10^6$ \gls{OFDM} symbols were randomly generated, and the inverse discrete Fourier transform of each symbol was taken to obtain time domain symbols.
The \gls{PAPR} \gls{CDF} was generated based on these time domain symbols.
As one can see, \gls{QAM}-\gls{SIP} and \gls{GS} lead to nearly the same distributions of the \gls{PAPR} as conventional \gls{QAM}.
From these results, these two approaches should not lead to higher distortion of the transmitted signal than conventional \gls{QAM}.
This is encouraging towards implementation and use of the proposed schemes.

\section{Conclusion}
\label{sec:conclu}

We have evaluated the performance of a neural receiver considering an \gls{OFDM} channel model that includes frequency selectivity and channel aging.
Evaluations were performed for different speed ranges.
Our results show that when a few orthogonal pilots are used, a neural receiver operating over multiple subcarriers and \gls{OFDM} symbols enables significantly lower \glspl{BER}.
The gains are more pronounced in high mobility scenarios.
We have then demonstrated how joint optimization of the transmitter and neural receiver allows reliable symbol detection without the need for orthogonal pilots.
This permits throughput gains as no \gls{RE} is wasted for the transmission of reference signals.
Suppression of orthogonal pilots can be achieved either by learning \glspl{SIP} that are linearly combined with \gls{QAM} modulated data carrying symbols, or by learning a zero mean constellation that is used to modulate the data.
Because the learned constellation is forced to have zero mean, the second approach does not leverage any form of reference signal (orthogonal or superimposed).
Moreover, our simulations reveal that both schemes do not negatively affect the \gls{PAPR} compared to conventional \gls{QAM}.
We therefore believe that such schemes could be part of beyond-5G communication systems, as they allow unprecedented throughput and reliability, while removing the need for \glspl{DMRS}.
Apart from throughput gains, pilotless transmissions could remove the control signaling overhead related to the choice of the best suitable pilot pattern, as the learned constellations and \gls{SIP} patterns work for any \gls{SNR}, Doppler, or delay spread.
Future work could include extending this approach to \gls{MIMO} systems, where orthogonal pilots could still be needed to estimate the channel of the different users and compute the equalization and precoding matrices.
Another line of work could be the extension of such schemes to other channel models, e.g., with subcarrier interference or no cycle prefix.

\begin{appendices}

\section*{Appendix}

\subsection{Derivation of the IEDD baseline}
\label{app:iedd_der}

\subsubsection*{Derivation of~(\ref{eq:iedd_prs})}

We denote by $\text{LLR}_P(k,i)$ the prior information assumed to be available to the estimator for each bit $i \in \{1,\dots,m\}$ of each resource element $k \in \{1,\dots,n\}$.
Let $P_{X_k}$ be the corresponding prior distribution over the transmitted data symbols on the $k^{th}$ resource element.
For a given constellation point $c_u \in \Cc$, $1 \leq u \leq 2^m$, with corresponding labeling $\LB c_u^{(1)},\dots,c_u^{(m)} \RB $, with $c_u^{(i)} \in \{0,1\}$, $1 \leq i \leq m$, we have
\begin{align}
	\ln{P_{X_k}(c_u)} & = \ln{P \LB B_{k,1} = c_u^{(1)},\dots, B_{k,m} = c_u^{(m)} \RB } \nonumber\\
			   		 & \approx \sum_{i=1}^m \ln{P \LB B_{k,i} =  c_u^{(i)} \RB } \label{eq:ind_bit1}
\end{align}
where $B_{k,i}$ is the random variable corresponding to the $i^{th}$ bit transmitted in the $k^{th}$ resource element, and~(\ref{eq:ind_bit1}) is exact assuming the bits mapped to a same resource element are independent.
However, because the prior information corresponds to the output of the decoder, it is conditioned on the channel output, making the bit levels not independent in general.
Because the computation of the exact prior of the symbols is not possible from the marginal distributions of the bits only,~(\ref{eq:ind_bit1}) is only an approximation.

Moreover, we have
\begin{align}
	\ln{P \LB B_{k,i} = 0 \RB} &= -\ln{1 + \frac{P \LB B_{k,i} = 1 \RB}{P \LB B_{k,i} = 0 \RB}} \nonumber\\
								 &= -\ln{1 + \exp{\text{LLR}_P(k,i)}} \label{eq:llr_p0}
\end{align}
where the last equality comes from the definition of the \gls{LLR} as $\text{LLR}_P(k,i) \coloneqq \ln{\frac{P \LB B_{k,i} = 1 \RB}{P \LB B_{k,i} = 0 \RB}}$.
Similarly, we have
\begin{equation}
	\label{eq:llr_p1}
	\ln{P \LB B_{k,i} = 1 \RB} = \text{LLR}_P(k,i) - \ln{1 + \exp{\text{LLR}_P(k,i)}}.
\end{equation}
By combining~(\ref{eq:ind_bit1}),~(\ref{eq:llr_p0}) and~(\ref{eq:llr_p1}), we get
\begin{equation}
	\label{eq:log_pr}
	\ln{P_{X_k}(c_u)} \approx \sum_{i=1}^m c_u^{(i)} \text{LLR}_P(k,i) - \sum_{i=1}^m \ln{1 + \exp{\text{LLR}_P(k,i)}}.
\end{equation}
The second term in the left-hand side of~(\ref{eq:log_pr}) does not depend on $c_u$.
Therefore, the vector $\LSB \sum_{i=1}^m c_1^{(i)} \text{LLR}_P(k,i),\cdots,\sum_{i=1}^m c_{2^m}^{(i)} \text{LLR}_P(k,i)  \RSB$ corresponds to unscaled log-probabilities, from which the distribution $P_{X_k}$ can be recovered using the $\text{softmax}\LB\cdot\RB$ function:
\begin{equation}
\label{eq:pr}
\LSB P_{X_k}(c_1),\dots,P_{X_k}(c_{2^m}) \RSB \approx \text{softmax}\LB \sum_{i=1}^m c_1^{(i)} \text{LLR}_P(k,i),\cdots,\sum_{i=1}^m c_{2^m}^{(i)} \text{LLR}_P(k,i)  \RB
\end{equation}
with
\begin{equation}
\text{softmax}\LB l_1, \dots, l_{2^m} \RB \coloneqq  \LSB \frac{\exp{l_1}}{\sum_{i=1}^{2^m} \exp{l_i}}, \dots, \frac{\exp{l_{2^m}}}{\sum_{i=1}^{2^m} \exp{l_i}} \RSB.
\end{equation}

\subsubsection*{Derivation of~(\ref{eq:iedd_est}) and~(\ref{eq:est_err_iedd})}

Equations~(\ref{eq:iedd_est}) and~(\ref{eq:est_err_iedd}) follow directly from applying the \gls{LMMSE} estimator.
We start by rewriting the \gls{OFDM} channel transfer function~(\ref{eq:ch_tr}) as
\begin{equation*}
	\yv = \text{diag}\LB \xv \RB \hv + \wv
\end{equation*}
where $\yv = \text{vec}\LB \Ym \RB$,  $\xv = \text{vec}\LB \Xm \RB$, $\hv = \text{vec}\LB \Hm \RB$, and $\wv = \text{vec}\LB \Wm \RB$.
We remind that $\EE \LP \hv\hv\htp \RP = \Rm$ and $\EE \LP \wv\wv\htp \RP = \sigma^2\Id_n$.
Moreover, it is assumed that $x_k$ is distributed according to the prior $P_{X_k}$, $1 \leq k \leq n$.
The \gls{LMMSE} channel estimate of $\hv$ is (e.g.,~\cite[Lemma B.17]{massivemimobook})
\begin{equation}
\label{eq:lmmse_gen}
\widehat{\hv} = \EE \LP \hv \yv\htp \RP \EE \LP \yv \yv\htp \RP^{-1} \yv
\end{equation}
and the covariance matrix of the estimation error is
\begin{equation}
\widetilde{\Rm} = \EE \LP \hv \hv\htp \RP - \EE \LP \hv \yv\htp \RP \EE \LP \yv \yv\htp \RP^{-1} \EE \LP \yv \hv\htp \RP = \Rm - \EE \LP \hv \yv\htp \RP \EE \LP \yv \yv\htp \RP^{-1} \EE \LP \yv \hv\htp \RP.\label{eq:lmmse_err_gen}
\end{equation}
Then,
\begin{equation}
	 \EE \LP \hv \yv\htp \RP = \EE \LP \hv\hv\htp \RP \EE \LP \xv\htp \RP = \Rm \text{diag}\LB \bar{\xv} \RB\htp \label{eq:iedd_lmmse_1}
\end{equation}
where $\bar{x}_k \coloneqq \EE \LP x_k \RP$ with the expectation taken according to the prior $P_{X_k}$.
Moreover,
\begin{align}
	 \EE \LP \yv \yv\htp \RP &= \EE \LP \LB \text{diag}\LB\xv\RB\hv \RB \LB \text{diag}\LB\xv\RB\hv \RB\htp \RP + \sigma^2\Id_n \nonumber\\
	 &= \EE \LP \LB \xv \circ \hv \RB \LB \xv \circ \hv \RB\htp \RP + \sigma^2\Id_n \nonumber\\
	 &= \EE \LP \LB \xv\xv\htp \RB \circ \LB \hv \hv\htp \RB \RP + \sigma^2\Id_n \nonumber \\
	 &= \EE \LP \xv\xv\htp \RP \circ \Rm + \sigma^2\Id_n \label{eq:iedd_lmmse_2}
\end{align}
and
\begin{equation}
	 \EE \LP \yv \hv\htp \RP = \text{diag}\LB \bar{\xv} \RB \Rm\htp. \label{eq:iedd_lmmse_3}
\end{equation}

Combining~(\ref{eq:lmmse_gen}),~(\ref{eq:lmmse_err_gen}),~(\ref{eq:iedd_lmmse_1}),~(\ref{eq:iedd_lmmse_2}), and~(\ref{eq:iedd_lmmse_3}) leads to the desired results:
\begin{equation*}
\widehat{\hv} = \Rm \text{diag}\LB \bar{\xv} \RB\htp \LB \EE \LP \xv\xv\htp \RP \circ \Rm + \sigma^2\Id_n \RB^{-1} \yv
\end{equation*}
and
\begin{equation*}
\widetilde{\Rm} = \Rm - \Rm \text{diag}\LB \bar{\xv} \RB\htp \LB \EE \LP \xv\xv\htp \RP \circ \Rm + \sigma^2\Id_n  \RB^{-1} \text{diag}\LB \bar{\xv} \RB \Rm\htp.
\end{equation*}

\subsubsection*{Derivation of~(\ref{eq:iedd_llr})}

Focusing on the $k^{th}$ resource element, $1 \leq k \leq n$, we can rewrite the channel transfer function as
\begin{equation}
y_k = \widehat{h}_k x_k + \underbrace{\widetilde{h}_k x_k + w_k}_{\widetilde{w}_k}
\end{equation}
where $w_k \sim \Cc\Nc(0, \sigma^2)$, $\widehat{h}_k$ is the \gls{LMMSE} estimate of the channel response $h_k$, and $\widetilde{h}_k$ the channel estimation error with variance $\widetilde{R}_{k,k}$.
It is assumed that prior information is available on the transmitted bits in the form of \glspl{LLR}, which for the $i^{th}$ bit transmitted in the $k^{th}$ resource element, is denoted by $\text{LLR}_E(k,i)$.
The transmitted symbols $X_k$ are assumed to be \emph{apriori} distributed according to a distribution $\widetilde{P}_{X_k}$ computed as in~(\ref{eq:iedd_prs}) but using the extrinsic information $\text{LLR}_E$ instead of $\text{LLR}_P$.
Note that $\EE\LB \widetilde{w}_k\widetilde{w}_k^* \RB = \widetilde{R}_{k,k} + \sigma^2$.
Given a constellation point $c \in \Cc$ and assuming $\widetilde{w}_k$ is Gaussian distributed we have
\begin{align}
	P \LB X_k = c | \widehat{h}_k, y_k \RB &= \frac{\widetilde{P}_{X_k} \LB c \RB P \LB y_k|X_k = c, \widehat{h}_k  \RB}{P \LB y_k|\widehat{h}_k \RB }\\
		&= \frac{\exp{-\frac{1}{\widetilde{\sigma_k}'^2}\abs{y_k - \widehat{h}_k c}^2 + \sum_{i=1}^m c^{(i)} \text{LLR}_E(k,i)}}{P \LB y_k|\widehat{h}_k \RB \pi\widetilde{\sigma_k}'^2 \sum_{a \in \Cc}\exp{\sum_{i=1}^m a^{(i)} \text{LLR}_E(k,i)}} 		
\end{align}
where $\widetilde{\sigma_k}'^2 = \widetilde{R}_{k,k} + \sigma^2$.
Now, we have
\begin{align}
\text{LLR}(k,i) &= \ln{\frac{P \LB B_{k,i} = 1 | \widehat{h}_k, y_k \RB}{P \LB B_{k,i} = 0 | \widehat{h}_k, y_k \RB}} \nonumber\\
				&= \ln{\frac{\sum_{c \in \Cc_{i,1}} \exp{-\frac{1}{\widetilde{\sigma_k}'^2}\abs{y_k - \widehat{h}_k c}^2 + \sum_{i=1}^m c^{(i)} \text{LLR}_E(k,i)}}{\sum_{c \in \Cc_{i,0}} \exp{-\frac{1}{\widetilde{\sigma_k}'^2}\abs{y_k - \widehat{h}_k c}^2 + \sum_{i=1}^m c^{(i)} \text{LLR}_E(k,i)}}}
\end{align}
where $\Cc_{i,0}$($\Cc_{i,1}$) is the subset of $\Cc$ which contains all constellation points with the $i^{th}$ labeling bit set to 0 (1).
A similar expression can be found in~\cite{775793}.

\subsection{Proof of Proposition~\ref{prop:ach}}
\label{app:rate_proof}

We provide here a proof that~(\ref{eq:rate}) is an achievable rate assuming some conditions on the channel and receiver hold true.
The stochastic process generating the channel response is assumed to be memoryless and weak-sense stationary.
By weak-sense stationary, we mean that for any couple of indexes $(i,k)$, $\EE \LP h_i \RP = \EE \LP h_k \RP$, $\EE \LP h_i h_k^* \RP = \EE \LP h_{i-k} h_0^* \RP$, and $\EE \LP \abs{h_i}^2 \RP < \infty$.
The channel model presented in Section~\ref{sec:ch_mod} is both memoryless and stationary.
Moreover, by achievable rate, we mean according to the standard definition~\cite[§8.5]{cover1999elements}.

We leverage the random coding argument.
We aim to transmit a message randomly chosen from a set $\{1,\dots,2^u\}$, where $u$ is an integer.
Each resource element takes as input $m$ bits, and a frame is formed of $n$ resource elements, among which a set $\Nc_D$ of size $n_D$ can be used to transmit data.
A code $\Bc$ is constructed by generating $2^{u}$ binary codewords of length $v n_D m$ denoted by $\bv$, where $v$ is the number of frames over which a codeword spreads.\footnote{A codeword is contained in a single frame in most practical systems.
However, for the sake of this proof, we allow a codeword to spread over multiple frames.}
Codewords are generated by independently and uniformly drawing bits.
A message $w \in \{1,\dots,2^{u}\}$ is encoded at the transmitter by mapping it to a codeword $\bv(w) \in \Bc$ in a manner known to the receiver.
For convenience, bits forming a codeword $\bv$ are indexed by $b_{l,k,i}$ to refer to the $i^{th}$ bit on the $k^{th}$ resource element of the $l^{th}$ frame.

On the receiver side, reconstruction of the transmitted codeword is achieved using for each bit $i$ and resource element $k$ a non-negative decoding metric that jointly operates on the $n_D$ data carrying resource elements of a frame of received symbols taken from the channel output alphabet.
This metric is denoted by $q_{k,i} \LB b, \yv \RB $ where $b$ is either $0$ or $1$, and $\yv$ is a vector of received data symbols of length $n_D$ corresponding to a single frame.
Decoding is performed by selecting the codeword $\widehat{w}$ such that
\begin{equation}
\widehat{w} = \argmax{w} \LB \prod_{l=1}^v \prod_{k \in \Nc_D} \prod_{i=1}^m q_{k,i}\LB b_{l,k,i}(w), \yv_l \RB \RB
\end{equation}
where $\yv_l$ is the vector of channel output symbols corresponding to the $l^{th}$ frame.
We introduce for convenience
\begin{equation}
\label{eq:qbar}
\bar{q}(\bv, \yv) \coloneqq \prod_{l=1}^v \prod_{k \in \Nc_D} \prod_{i=1}^m q_{k,i}\LB b_{l,k,i}, \yv_l \RB.
\end{equation}

Following the approach adopted in~\cite{bmi,arxiv.1205.1389}, we fix the transmitted message to $w$, the channel input $\mathbf{\mathcal{X}}$ to $\bv(w)$, and the channel output $\mathbf{\mathcal{Y}}$ to $\yv = \LSB \yv_1,\dots,\yv_v \RSB$, and compute the probability that a decoding error occurs, i.e.,
\begin{equation}
P \LB E | \mathbf{\mathcal{X}} = \bv(w), \mathbf{\mathcal{Y}} = \yv \RB =
P \LB \bigcup_{w' \neq w} \LSB \bar{q} \LB \bv(w'), \yv \RB \geq \bar{q} \LB \bv(w), \yv \RB \RSB \middle| \mathbf{\mathcal{X}} = \bv(w), \mathbf{\mathcal{Y}} = \yv \RB.
\end{equation}
Note that randomness comes from the random construction of the code $\Bc$.
We have
\begin{align}
P \LB E | \mathbf{\mathcal{X}} = \bv(w), \mathbf{\mathcal{Y}} = \yv \RB &\leq \sum_{w' \neq w} P \LB \bar{q} \LB \bv(w'), \yv \RB \geq \bar{q} \LB \bv(w), \yv \RB \middle| \mathbf{\mathcal{X}} = \bv(w), \mathbf{\mathcal{Y}} = \yv \RB \label{eq:ub}\\
	&\leq \frac{1}{\bar{q} \LB \bv(w), \yv \RB} \sum_{w' \neq w} \EE_{\Bc} \LP \bar{q} \LB \bv(w'), \yv \RB \middle| \mathbf{\mathcal{Y}} = \yv\RP \label{eq:markov}\\
	&\leq \frac{2^u}{\bar{q} \LB \bv(w), \yv \RB} \EE_{\bv} \LP \bar{q} \LB \bv, \yv \RB \middle| \mathbf{\mathcal{Y}} = \yv \RP \label{eq:iid_cw} \\
	&\leq 2^u \frac{\EE_{\bv} \LP \prod_{l=1}^v \prod_{k \in \Nc_D} \prod_{i=1}^m q_{k,i}\LB b_{l,k,i}, \yv_l \RB \middle| \mathbf{\mathcal{Y}} = \yv \RP}{\prod_{l=1}^v \prod_{k \in \Nc_D} \prod_{i=1}^m q_{k,i}\LB b_{l,k,i}(w), \yv_l \RB} \label{eq:def_qbar} \\
	&\leq 2^u \prod_{l=1}^v \prod_{k \in \Nc_D} \prod_{i=1}^m \frac{\EE_{b_{l,k,i}} \LP  q_{k,i}\LB b_{l,k,i}, \yv_l \RB \RP}{q_{k,i}\LB b_{l,k,i}(w), \yv_l \RB} \label{eq:ind_bit}
\end{align}
where~(\ref{eq:ub}) follows from the union bound,~(\ref{eq:markov}) follows from the Markov inequality,~(\ref{eq:iid_cw}) follows from the uniformly and independently drawn codewords,~(\ref{eq:def_qbar}) follows from~(\ref{eq:qbar}), and~(\ref{eq:ind_bit}) follows from the assumption that bits forming codewords are independently and uniformly drawn.
We can rewrite~(\ref{eq:ind_bit}) as
\begin{equation}
P \LB E | \mathbf{\mathcal{X}} = \bv(w), \mathbf{\mathcal{Y}} = \yv \RB \leq 2^{-v\LB T_v - R_c \RB} \label{eq:exp}
\end{equation}
where $R_c \coloneqq \frac{u}{v}$ is the rate [\si{\bit\per\frame}] and
\begin{align}
T_v(w,\yv) &\coloneqq \frac{1}{v}\sum_{l=1}^v \sum_{k \in \Nc_D} \sum_{i=1}^m \log{\frac{q_{k,i} \LB b_{l,k,i}(w), \yv_l \RB}{\EE_{b_{l,k,i}} \LP q_{k,i}\LB b_{l,k,i}, \yv_l \RB \RP}}\\
	&= n_Dm + \sum_{k \in \Nc_D} \sum_{i=1}^m \frac{1}{v} \sum_{l=1}^v \log{\frac{q_{k,i}\LB b_{l,k,i}(w), \yv_l \RB}{q_{k,i}\LB 0, \yv_l \RB + q_{k,i}\LB 1, \yv_l \RB}} \label{eq:Tv}.
\end{align}
One can notice that
\begin{equation}
Q_{k,i} \LB b | \yv_l \RB \coloneqq \frac{q_{k,i}\LB b, \yv_l \RB}{q_{k,i}\LB 0, \yv_l \RB + q_{k,i}\LB 1, \yv_l \RB}
\end{equation}
forms a probability mass function over the two possible outcomes of the $i^{th}$ bit carried by the $k^{th}$ resource element.
Because the channel process is Gaussian and weak-sense stationary, it is also strict-sense stationary~\cite[Theorem~4.2]{lindgren02}.
Moreover, because the covariance function~(\ref{eq:time_cov}) of the channel process converges to 0 as $\abs{i-k}$ goes to infinity, the channel process is also ergodic~\cite[Lemma~5.2]{lindgren02}.
The additive noises and codewords are \gls{iid} and therefore also are ergodic processes.
Therefore, we can apply the mean ergodic theorem~\cite[Theorem~5.4]{lindgren02}:
\begin{equation}
P \LB \lim_{v\to\infty} \frac{1}{v} \sum_{l=1}^v \log{Q_{k,i} \LB b_{l,k,i} | \yv_l \RB} = \EE_{\yv',b_{k,i}} \LP \log{Q_{k,i} \LB b_{k,i}| \yv' \RB} \RP \RB  = 1.
\end{equation}
Note that randomness comes from random drawing of $b_{l,k,i}$ and $\yv_l$.
Therefore, almost surely, for any positive real number $\delta$, there exists $v$ large enough such that
\begin{equation}
\frac{1}{v} \sum_{l=1}^v \log{Q_{k,i} \LB b_{l,k,i} | \yv_l \RB} \geq \EE_{\yv',b_{k,i}} \LP \log{Q_{k,i} \LB b_{k,i}| \yv' \RB} \RP - \frac{\delta}{n_Dm}
\end{equation}
which, when combined with~(\ref{eq:exp}) and~(\ref{eq:Tv}), leads almost surely to
\begin{equation}
P \LB E \RB \leq 2^{-v\LB n_Dm + \sum_{k \in \Nc_D} \sum_{i=1}^m \EE_{\yv',b_{k,i}} \LP  \log{Q_{k,i} \LB b_{k,i}| \yv' \RB} \RP - \delta - R_c \RB}.
\end{equation}
One can conclude from this result that
\begin{equation}
R \coloneqq n_Dm + \sum_{k \in \Nc_D} \sum_{i=1}^m \EE_{\yv',b_{k,i}} \LP \log{Q_{k,i} \LB b_{k,i}| \yv' \RB} \RP
\end{equation}
is an achievable rate, i.e., for any $R_c < R$, the probability of error can be made arbitrarily low.
$R$ can be rewritten as
\begin{equation}
R = \sum_{k \in \Nc_D} \sum_{i=1}^m I(B_{k,i}, \mathbf{\mathcal{Y}}) - \sum_{k \in \Nc_D} \sum_{i=1}^m \EE_{\yv'} \LP \text{D}_{\text{KL}}\LB P_{B_{k,i}}\LB \cdot | \yv' \RB || Q_{k,i}\LB \cdot |\yv' \RB \RB \RP
\end{equation}
where $P_{B_{k,i}}(\cdot|\yv)$ is the true posterior distribution on the $i^{th}$ bit transmitted on the $k^{th}$ resource element conditioned on $\yv$.

\end{appendices}

\bibliographystyle{IEEEtran}
\bibliography{IEEEabrv,references}

\end{document}